\documentclass{article}
\usepackage{mathrsfs}
\usepackage{a4wide}
\usepackage{amsmath,amsfonts,amssymb,amsthm}
\usepackage{color}

\def \beq{\begin{equation}}
\def \eeq{\end{equation}}

\def\supp{\mathop{\rm supp} \nolimits} 

\def\and {{\rm \; and \;}}

\renewcommand{\Im}{{\rm Im\,}}

\renewcommand{\Re}{{\rm Re\,}}

\newcommand{\Green}{{\mathcal G}}
\newcommand{\R}{{\mathbb R}}

\newcommand{\y}{{\bf y}}
\newcommand{\C}{\mathbb{C}}
\newcommand{\x}{{\bf x}}

\newcommand{\A}{{\bf A}}

\DeclareMathOperator{\curl}{curl}
\DeclareMathOperator{\Div}{div}
\DeclareMathOperator{\Tr}{Tr}

\newtheorem{theorem}{Theorem}[section]

\newtheorem{proposition}[theorem]{Proposition}
\newtheorem{corollary}[theorem]{Corollary}
\newtheorem{lemma}[theorem]{Lemma}

\theoremstyle{definition}
\newtheorem{remark}[theorem]{Remark}

\numberwithin{equation}{section}

\begin{document}

\noindent 
\begin{center}
\textbf{\large Sharp trace asymptotics for a class of $2D$-magnetic operators}
\end{center}

\begin{center}
July 29, 2011
\end{center}

\vspace{0.5cm}

\begin{center}
Horia D. Cornean\footnote{Department of Mathematical Sciences, 
    Aalborg
    University, Fredrik Bajers Vej 7G, 9220 Aalborg, Denmark; e-mail:
    cornean@math.aau.dk},
S{\o}ren Fournais\footnote{Department of Mathematical Sciences, University of Aarhus, Ny Munke\-gade,
Building 1530,
DK-8000 Aarhus C, Denmark}${}^,$\footnote{on leave from CNRS and Laboratoire de Math\'{e}matiques
  d'Orsay, Univ Paris-Sud, Orsay CEDEX, F-91405; e-mail: fournais@imf.au.dk},
Rupert L. Frank\footnote{Department of Mathematics, Fine Hall, 
Princeton University
Princeton, NJ 08544 USA; \\ e-mail: rlfrank@math.princeton.edu}, 
Bernard Helffer\footnote{Laboratoire de Math{\'e}matiques, 
Univ Paris Sud et CNRS, B{\^a}timent 425,
91405 Orsay Cedex - France ; e-mail: Bernard.Helffer@math.u-psud.fr}
     
\end{center}

\noindent

\begin{abstract}
In this paper we prove a two-term asymptotic formula for for the spectral counting function for a $2$D magnetic Schr\"{o}dinger operator on a domain (with Dirichlet boundary conditions) in a semiclassical limit and with strong magnetic field. By scaling, this is equivalent to a thermodynamic limit of a $2$D Fermi gas submitted to a constant external magnetic field.

The original motivation comes from a paper by H. Kunz in which he studied, among other things, the boundary correction for the grand-canonical pressure and density of such a Fermi gas. Our main theorem yields a rigorous proof of the formulas 
announced by Kunz. Moreover, the same theorem provides several other results on the integrated density of states for operators of the type $(-ih\nabla- \mu {\bf A})^2$ in $L^2({\Omega})$ with Dirichlet boundary conditions.  
\end{abstract}

\vspace{0.5cm}

\tableofcontents

\section{Introduction and the main results}

\subsection{The setting and the semi-classical results}

For $h>0$ we
consider 
the self-adjoint operator
$$
L_h = (-ih\nabla - \A)^2 \qquad\text{in}\ L^2(\Omega) \,,
$$
where $\Omega\subset\R^2$ is a smooth, bounded and connected 
domain, and $\A$ is a
smooth, 
real vector field defined on $\overline\Omega$. Its curl, $B:=\curl\A$, 
describes a magnetic field (which in the case of a simply-connected 
$\Omega$ determines $\A$ uniquely up to the addition of a gradient
field). Throughout we shall assume that
\begin{equation}
\label{eq:bpos}
\inf_{\x\in\Omega} B(\x) >0 \,.
\end{equation}
Our goal is to study the behavior of the trace $\Tr f(h^{-1} L_h)$ 
in the limit $h\to 0$ for any Schwartz function $f$. Our main result 
will be a two-term asymptotic formula
\begin{equation}\label{eq:asymp}
\Tr f(h^{-1} L_h) = h^{-1} \left( C_0(f) + h^{1/2} C_1(f) + o(h^{1/2}) \right)
\qquad\text{as}\ h\to 0
\end{equation}
with explicit coefficients $C_0(f)$ and $C_1(f)$. The terms 
$C_0(f)$ and $C_1(f)$ are bulk and boundary terms, respectively, 
and will be defined below.

Due to the special structure of $C_0(f)$,  asymptotics
\eqref{eq:asymp} 
will imply two-term asymptotics for the number of eigenvalues of $L_h$ 
less than $E h$, provided the energy $E$ satisfies a certain gap
assumption. 
In the particular case of a constant magnetic field, 
this assumption is satisfied for all energies different from 
the Landau levels. Thus we are looking for a magnetic analog of the 
formula \cite{Iv}
\begin{equation}
 \label{eq:weyl}
N(E h, -h^2\Delta) = (4\pi h)^{-1} \left( |\Omega| \,E - h^{1/2}\, |\partial\Omega|\, E^{1/2} + o(h^{1/2}) \right) \qquad\text{as}\ h\to 0\,,
\end{equation}
where $-\Delta$ is the Dirichlet Laplacian on $\Omega$ and, 
for any self-adjoint operator $A$ with discrete spectrum, 
$N(\lambda,A)$ denotes the number of eigenvalues less than $\lambda$,
counting multiplicities. We emphasize that, in contrast to our result
in 
the magnetic case, \eqref{eq:weyl} requires some global geometric
conditions on $\Omega$. Moreover, it is remarkable that our proof of a
two-term 
asymptotics of the counting function is technically much 
less involved than the proof of \eqref{eq:weyl}.

We emphasize that in \eqref{eq:asymp} we consider $f(h^{-1} L_h)$
instead of $f(L_h)$, that is, we look at eigenvalues of order $h$
instead of order $1$. (For semi-classical results about $\Tr f(P_h)$
and $N(E,P_h)$ for a large class of operators $P_h$ we refer to
\cite{HeRo,DiSj}. These operators $P_h$ are defined on the whole space,
however, so that there are no boundary effects.) The study of $\Tr
f(h^{-1} L_h)$ has applications to mean-field and thermodynamic limits
in physical problems. It appears in \cite{LSY} in connection with the
mathematical justification of mean field theories for a large number
of electrons in strong magnetic fields. Another motivation for us
comes from the paper \cite{Ku} which studies a Fermi gas confined to a
large convex set submitted to a constant magnetic field. Boundary
corrections are derived for the grand-canonical pressure,
magnetization and density in the thermodynamic limit. Our main result
\eqref{eq:asymp} provides a rigorous proof of this derivation and 
removes the convexity assumption. 

We next describe the coefficients entering in \eqref{eq:asymp}. 
The term $C_0(f)$ is a bulk term, in the sense that it is of the form
\begin{equation}
 \label{eq:c0}
C_0(f) := (2\pi)^{-1} \sum_{k=1}^\infty \int_\Omega b_k(B(\x),f) B(\x) \,d\x
\end{equation}
involving the integral over $\Omega$ of the density
\begin{equation}
 \label{eq:c0den}
b_k(B,f) := f((2k-1)B) \,.
\end{equation}
We emphasize that the leading order asymptotics 
$\Tr f(h^{-1} L_h) \sim h^{-1} C_0(f)$ are essentially known. 
Indeed, related formulas have appeared in \cite{CdV,Ta} in the 
case of magnetic fields growing at infinity and in 
\cite{LSY,So1,So2,So3} in the presence of decaying electric potentials. 

Our main result concerns the second term. This is a surface term of the form
\begin{equation}
 \label{eq:c1}
C_1(f) := (2\pi)^{-1} \sum_{k=1}^\infty \int_{\partial\Omega} s_k(B(\x),f) \sqrt{B(\x)} \,d\sigma(\x) \,,
\end{equation}
depending only on the restriction of $B$ to the boundary. Here $d\sigma$ denotes integration with respect to the surface measure. In order to define the density $s_k(B,f)$, we recall that the harmonic oscillator
$$
-\frac{d^2}{dt^2}+t^2 \qquad\text{in}\ L^2(\R) 
$$
has eigenvalues $2k-1$, $k=1,2,\ldots,$ with corresponding normalized eigenfunctions
\begin{equation}\label{modello2}
\phi_k(t):=\frac{1}{2^{k/2} \sqrt{k!}}(1/\pi)^{\frac{1}{4}}H_k(t)e^{-\frac{t^2}{2}},
\end{equation}
where $H_k$ is the $k$-th Hermite polynomial. For $\xi\in\R$ we consider the Dirichlet realization of the operator
\begin{equation}\label{modello}
-\frac{d^2}{dt^2}+(\xi+t)^2 \qquad\text{in}\ L^2(\R_+) \,.
\end{equation}
This operator has discrete spectrum with simple eigenvalues $e_k(\xi)$, $k=1,2,\ldots$. Let $\psi_k(\cdot,\xi)$ denote corresponding normalized eigenfunctions chosen to be positive in a neighborhood of $+\infty$. (This is possible since by the Sturm oscillation theorem eigenfunctions have only a finite number of zeros.) We put
\begin{equation}
 \label{eq:c1den}
s_k(B,f) := \int_0^\infty  \left( \int_\R f(B e_k(\xi)) |\psi_k(t,\xi)|^2 d\xi - f(B(2k-1)) \right) \, dt \,.
\end{equation}
Here and in the following, by a Schwartz function on $[0,\infty)$ we
mean the restriction of a Schwartz function on $\R$ to $[0,\infty)$. 
We are now in position to state our main result. The coefficient $C_1(f)$ is well-defined according to the following

\begin{lemma}\label{salternative}
 Let $f$ be a Schwartz function on $[0,\infty)$. For $B>0$ the sum $\sum_{k=1}^\infty s_k(B,f)$ converges absolutely and is continuously differentiable with respect to $B>0$. Moreover, for any $k=1,2,\ldots,$
\begin{equation}
 \label{eq:salternative}
s_k(B,f) = \int_\R  \left( f(B e_k(\xi)) - f(B(2k-1)) \int_\xi^\infty |\phi_k(t)|^2 \,dt \right) \, d\xi \,.
\end{equation}
\end{lemma}

And here is the main theorem:

\begin{theorem}\label{theorem2}
 Let $\Omega\subset\R^2$ be a bounded and connected domain with a boundary $\partial\Omega$ assumed to be a finite union of disjoint, regular, simple and closed $C^2$-curves. Let $\A\in C^3(\overline\Omega,\R^2)$ and assume that $B:=\curl\A$ satisfies \eqref{eq:bpos}. Then the asymptotics \eqref{eq:asymp} are valid for any Schwartz function $f$ on $[0,\infty)$ and with $C_0(f)$ and $C_1(f)$ defined by \eqref{eq:c0}, \eqref{eq:c0den} and \eqref{eq:c1}, \eqref{eq:c1den}, respectively.
\end{theorem}

In order to state the promised result about the eigenvalue counting function, we recall that $N(\lambda,L_h)$ denotes the number of eigenvalues of $L_h$ less than $\lambda$, counting multiplicities.

\begin{corollary}\label{theorem1}
Let $\Omega$ and $\A$ be as in Theorem \ref{theorem2} and assume that for some positive integer $K$
$$
(2K-1)B_{\rm max}<(2K+1)B_{\rm min}\,,
$$
where $B_{\rm min} := \inf_{\x\in\Omega} B(\x)$ and $B_{\rm max} := \sup_{\x\in\Omega} B(\x)$. Then, for any $E\in ((2K-1)B_{\rm max},(2K+1)B_{\rm min})$,
\begin{align}
 \label{eq:counting}
N(Eh,L_h) = & \frac1{2\pi h} \left( K \int_{\Omega}B(\x) \,d\x \right. \\
& - h^{1/2} \sum_{k=1}^K \int_{\partial\Omega}
\int_0^\infty  \left( 1- \int_{\{\xi:\ B(\x) e_k(\xi)<E\}} |\psi_k(t,\xi)|^2 d\xi \right) \,dt
\sqrt{B(\x)} \,d\sigma(\x) \notag \\
&  + o(h^{1/2}) \Big) \notag
\end{align}
as $h\to 0$. In particular, if $\curl\A=B$ is constant, then \eqref{eq:counting} holds for any $E\in\R\setminus\{B(2k-1):\ k=1,2,\ldots\}$.
\end{corollary}

We emphasize again that the crucial point in Theorem \ref{theorem2} and Corollary \ref{theorem1} is the existence of a second term in the asymptotic expansion. It is remarkable that a second term exists without any assumption on periodic orbits, in contrast to the case $B=0$ \cite{Iv}. The fact that this second term is a boundary contribution involving a one-dimensional auxiliary problem is a typical feature of two-term asymptotics; see, e.g., \cite{SaVa}.
In \cite{Ku} it is shown, using the alternative form of $s_k(B,f)$ from Lemma \ref{salternative} and the behavior of zeroes of special functions, that for any $B>0$, the sum
$$
\sum_{k=1}^K \int_0^\infty  \left( 1- \int_{\{\xi:\ B e_k(\xi)<E\}} |\psi_k(t,\xi)|^2 d\xi \right) \,dt
$$
is a non-negative, decreasing function of $E\in ((2K-1)B,(2K+1)B)$. It diverges to $+\infty$ as $E\searrow B(2K-1)$ and vanishes as $E\nearrow B(2K+1)$. This means that asymptotically the effect of the Dirichlet boundary is to reduce the number of eigenvalues, as compared to a bulk prediction. This is, of course, consistent with the non-asymptotic result from \cite{FLW}, which states that $N(Eh,L_h) \leq (2\pi h)^{-1} K B |\Omega|$ for any $E\in ((2K-1)B,(2K+1)B]$ and any $h>0$, provided the magnetic field $B$ is constant and the domain $\Omega$ is tiling. Whether this extends to more general $B$ and $\Omega$ is an open problem. The corresponding problem for $B=0$ is the famous P\'olya conjecture.

Results similar to Theorem \ref{theorem2} and Corollary \ref{theorem1} are derived in \cite{Fra2,FoKa} in the 
case of Neumann instead of Dirichlet boundary conditions and 
for $f$ supported in $[0,B]$ (or $E<B$ in \eqref{eq:counting}). In
that case, 
however, $C_0(f)=0$, which makes the analysis considerably easier. 

A related interesting problem is the one in which the comparison
operator lives on a strip of finite width \cite{BHRS,BRS}. Enlarging
the discussion, there exist some works dealing with spectral asymptotics
for magnetic operators defined in the exterior of bounded domains,
both with Dirichlet and Neumann boundary conditions (see \cite{Push,Persson} and
references therein). For a heuristic approach to sharp asymptotics for magnetic systems we
refer the reader to \cite{HorSmil}. 

\subsection{The thermodynamic limit}

In the case of a homogeneous magnetic field, the semi-classical limit 
is equivalent to the thermodynamic limit. We will now show how to
recover and 
extend the results of \cite{Ku} as corollaries of Theorem
\ref{theorem2}. 
For a set $\Omega$ as in the previous subsection we define
\begin{equation}\label{prima9}
\Omega_h:=h^{-1/2}\, \Omega = \{\x\in\R^2:\; h^{1/2}\x\in \Omega\} \,.
\end{equation}
We fix a constant $B>0$, put $\A_0(\x) := \frac{B}{2}(-x_2,x_1)$ and consider for any $L>0$ the operator
$$
H^L:=(-i\nabla -\A_0)^2 \qquad\text{in}\ L^2(\Omega_{L^{-2}})
$$
with Dirichlet boundary conditions. Note that via the scaling
$\x\mapsto L^{-1}\x$, 
the operator $H^L$ is unitary equivalent to $h^{-1} L_h$ in
$L^2(\Omega)$ with 
$h=L^{-2}$. Hence Theorem \ref{theorem2} implies

\begin{corollary}
 \label{limterm}
Let $\Omega$ be as in Theorem \ref{theorem2}, let $\A_0(\x) = \frac{B}{2}(-x_2,x_1)$ for some $B>0$ and let $f$ be a Schwartz function on $[0,\infty)$. Then
\begin{align}\label{azecea3}
\frac{1}{|\Omega_{L^{-2}}|}\; \Tr f(H^L)
= \frac{B}{2\pi} \sum_{k= 1}^\infty b_k(B,f)
+ L^{-1} \ \frac{\sqrt{B} \ |\partial \Omega|}{2\pi \ |\Omega|} 
\sum_{k= 1}^\infty s_k(B,f) + o(L^{-1})
\end{align}
as $L\rightarrow +\infty$.
\end{corollary}

Note that this corollary is, in particular, valid for $f(E)=\left(e^{\beta(E-\mu)}+1\right)^{-1}$ and 
$f(E)=\ln(1+e^{-\beta(E-\mu)})$ for $\beta>0$ and $\mu\in\R$. These choices correspond to the Fermi-Dirac distribution and the grand-canonical pressure at chemical potential $\mu$ and inverse temperature $\beta$, respectively, which were considered in \cite{Ku}. In view of \eqref{eq:salternative}, Corollary \ref{limterm} provides a proof the formulas derived in \cite{Ku}.


\subsection{Outline of the paper}

Our method to deal with the boundary contribution is reminiscent of
the techniques developed in \cite{HelMo1,HelMo2,FoHel1, FHBook} in the
mathematical study of surface superconductivity.

As in many other problems dealing with extended magnetic fields, gauge
invariance and magnetic perturbation theory plays a crucial role (see
\cite{Cor0}-\cite{CN3}, \cite{He1}-\cite{HeSj} and references
therein).

Here is a short description of the content of this paper: 
\begin{itemize}

\item The proof of our main result, Theorem \ref{theorem2}, is given
  in Section \ref{sectiuneadoi}. The main idea is to write the trace
  involving the original operator as a sum of traces involving model
  operators plus an error of order $h^\infty$. These model operators
  are either defined on the whole plane, or on slightly perturbed
  half-planes. The $h$-asymptotics for the model
  operators are given in the previous three sections as independent results.     

\item In Section \ref{wholespace} we consider operators defined on the
  whole plane. The main result here is Proposition \ref{bulk}, in which
  we give a sharp $h$-expansion of the diagonal of the integral kernel of
  $f(H)$. 

\item In Sections \ref{halfspace} and \ref{phalfspace} we treat the local trace
  $h$-asymptotics for operators defined with Dirichlet boundary
  conditions on the half-plane and on a perturbed half-plane, respectively.
   and when we restrict them near the boundary. The main
  technical results are Corollaries \ref{bdrycor} and \ref{bdry2cor} which, unlike
  Proposition \ref{bulk}, are not formulated for integral kernels but
  for local traces. It would be interesting to
  prove a variant of these corollaries in terms of integral kernels. 
 
\item Section~\ref{sectiuneasase} provides the proof of our results on the
integrated density of states when the extra gap-condition is
satisfied. 

\item In the end we have three appendices containing various more or less
known results, adapted and formulated according to our needs. 

\end{itemize}

\section{Local asymptotics for operators on the whole plane}\label{wholespace}

In this section we consider functions of the operator $(-i\nabla -
\A)^2$ on $L^2(\R^2)$. From Appendix \ref{apendixD} we know that 
$f((-i\nabla - \A)^2)$ has a continuous integral kernel if $f$ 
is a Schwartz function. In this section we provide an approximation for the diagonal value of kernel of $f((-i\nabla - \A)^2)$ which depends only on the value of the magnetic field at the point where the kernel is evaluated. This approximation is good provided the magnetic field is slowly varying. The precise statement is the following.

\begin{proposition}\label{bulk}
Given $f\in\mathcal S(\R)$, $M>0$ and $N\geq 0$, there is a constant
$C>0$
 such that, for any $B\in C^2(\R^2)$ satisfying $\| B
 \|_\infty+\|\nabla B\|_\infty \leq M$, for any $\x_0\in\R^2$,
 $h\in(0,1]$ and $\ell \in[1,\infty)$ such that
\begin{equation}
 \label{eq:bulkass}
h^{-1/2} \sup_{|\x-\x_0|\leq \ell } |\nabla B(\x)| + h^{-1} \sup_{|\x-\x_0|\leq \ell} \|\mathrm{Hess}\ B(\x)\| \leq M
\end{equation}
and for any $\A,\A_0\in C^3(\R^2;\mathbb R^2)$ satisfying $\curl\A=B$ and $\curl\A_0\equiv B(\x_0)$ one has
\begin{equation}
 \label{eq:bulk}
\left| f( (-i\nabla - \A)^2 )(\x_0,\x_0) - f( (-i\nabla - \A_0)^2 )(\x_0,\x_0) \right| 
\leq C \, (h + \ell^{-N}) \,.
\end{equation}
\end{proposition}

Later, in our proof of Theorem \ref{theorem2}, we will apply this proposition with $\ell=h^{-\epsilon}$ for some $\epsilon>0$. Then, when $N$ is chosen large enough, we obtain an error of the order $\mathcal O(h)$. The crucial point for our proof of Theorem \ref{theorem2} is that this is $o(h^{1/2})$. We believe that the order $\mathcal O(h)$ is best possible. Note also that we do not require the magnetic field $B$ to be positive.

Before beginning with the proof, we will derive from the pointwise estimate in \eqref{eq:bulk} a weaker, integral estimate for $\Tr g f( (-i\nabla - \A)^2 )$ with $g$ of `small' support. To do so, we also recall that if $\curl\A_0\equiv B$ is a non-zero constant, then
\begin{equation}
 \label{eq:landaudiag}
f( (-i\nabla - \A_0)^2 )(\x_0,\x_0) = \frac{|B|}{2\pi} \sum_{k=1}^\infty f((2k-1)|B|) 
= \frac{|B|}{2\pi} \sum_{k=1}^\infty b_k(|B|,f)
\end{equation}
for all $\x_0\in\R^2$. It is easy to see that the right side of the above formula can be continuously extended to $B=0$ by setting it equal to $(4\pi)^{-1} \int_0^\infty f(\lambda)\,d\lambda$.

\begin{corollary}\label{bulkcor}
Given $f\in\mathcal S(\R)$, $M>0$ and $N\geq 0$, there is a constant
$C>0$ such that, for any $B\in C^2(\R^2)$ satisfying
$\| B \|_\infty+\|\nabla B\|_\infty \leq M$, for any $\x_0\in\R^2$,
$h\in(0,1]$ and $\ell \in[1,\infty)$ satisfying \eqref{eq:bulkass},
  for any $g\in L^\infty(\{|\cdot-\x_0|\leq \ell /2\})$ and for any $\A\in
  C^3(\R^2;\mathbb R^2)$ with $\curl\A=B$,  one has
\begin{equation}
 \label{eq:bulkcor}
\left| \Tr\left[ g f( (-i\nabla - \A)^2 ) \right]
- \frac{1}{2\pi} \sum_{k=1}^\infty \int_{\R^2} g(\x) b_k(|B(\x)|,f) |B(\x)| \,d\x \right|
\leq C \, \|g\|_\infty \, \ell^2 \ (h + \ell^{-N}) \,.
\end{equation}
\end{corollary}

Indeed, this is an immediate consequence of \eqref{eq:landaudiag} 
and \eqref{eq:bulk} noting that \eqref{eq:bulkass} holds for 
any $\x\in\supp g$ with $\ell$ replaced by $\ell/2$.



We now turn to the proof of Proposition \ref{bulk}. We shall use the functional calculus based on the Helffer--Sj\"ostrand formula, which we recall next. Assume that $f\in\mathscr{S}(\R)$ is real-valued. For any given $N\geq 1$, we can construct an almost
analytic extension $f_{a,N}$ which obeys the following conditions \cite{He-Sjhva}:
\begin{enumerate}\label{conditiihesj}
\item $f_{a,N}(z)=f(z),\quad \forall z\in\R$;
\item ${\rm supp}(f_{a,N})\subset
\{z\in\mathbb{C}:\; |\Im z|<1\} =:\mathcal D$;
\item $f_{a,N}\in \mathscr{S}(\mathcal D)$;
\item there exists $C_N>0$ such that for  all $z\in\mathcal D$
\begin{equation}\label{apatra2}
 \left |\frac{\partial f_{a,N}}{\partial \overline{z}}(z)\right 
|\leq C_N \frac{|\Im z|^{\mathit N}}{\langle\Re z\rangle^{\mathit N}} \,.
\end{equation}
\end{enumerate}
Let $H$ be any self-adjoint operator. Assuming that $N\geq 2$, the Helffer-Sj\"ostrand formula reads:
\begin{equation}\label{azecea6}
f(H)=\frac{1}{\pi}\int_{\mathcal{D}} \frac{\partial
  f_{a,N}}{\partial \overline{z}}(H-z)^{-1}\,dxdy,\quad \mbox{ with }  z=x+iy\,. 
\end{equation}

The proof of the following lemma is elementary.

\begin{lemma}\label{deltaa}
Assume that $B\in C^2(\R^2)$ satisfies $\| B \|_\infty+\|\nabla B\|_\infty \leq M$ and
$$
h^{-1/2} \sup_{|\x|\leq \ell} |\nabla B(\x)| + h^{-1} \sup_{|\x|\leq \ell} \|\mathrm{Hess}\ B(\x)\| \leq M
$$
for some positive $M, h, \ell$. We define
\begin{equation}\label{eq:a}
\A(\x) := \int_0^1 dt \,t \, B(t\x) \, (-x_2,x_1)\,,
\qquad
\A_0(\x) := \frac{B(0)}2\, (-x_2,x_1) \,.
\end{equation}
and $\delta\A:=\A-\A_0$. Then for all $\x\in\R^2$
\begin{equation}\label{eq:deltaaglob}
|\delta \A(\x)| \leq M |\x|
\qquad\text{and}\qquad
|\Div\, \delta \A(\x)| \leq \frac13 M |\x| \,.
\end{equation}
Moreover, for all $\x$ with $|\x|\leq \ell$
\begin{equation}\label{eq:deltaaloc}
|\delta \A(\x)| \leq \frac 13 M h^{1/2} |\x|^2
\qquad\text{and}\qquad
|\Div\delta \A(\x)| \leq \frac 13 M h^{1/2} |\x|
\end{equation}
and
\begin{equation}\label{eq:deltaataylor}
\left|\delta\A(\x) - \frac 13 (\nabla B(0)\cdot\x) \,(-x_2,x_1) \right| \leq \frac18 M h |\x|^3 \,.
\end{equation}
\end{lemma}

\begin{proof}
The first bound in \eqref{eq:deltaaglob} follows by writing
\begin{equation}\label{eq:deltaarep}
\delta\A (\x) = \int_0^1 dt\, t \left(B(t\x)-B(0)\right) (-x_2,x_1)
\end{equation}
and estimating $|B(t\x)-B(0)|\leq 2M$. Similarly, for the second bound we write
\begin{equation}
 \label{eq:deltaadivrep}
\Div\, \delta\A (\x) = \int_0^1 dt\, t^2 \, \nabla B(t\x) \cdot (-x_2,x_1)
\end{equation}
and estimate $|\nabla B(t\x)| \leq M$.

For the first bound in \eqref{eq:deltaaloc} we use that $B(t\x)-B(0) = \int_0^t ds\, \nabla B(s\x)\cdot \x$ in \eqref{eq:deltaarep} to obtain
\begin{equation}\label{eq:deltaarep2}
\delta\A (\x) = \frac12 \int_0^1 ds\, (1-s^2) \left[\nabla B(s\x)\cdot \x\right] \, (-x_2,x_1) \,.
\end{equation}
The claimed inequality then follows from $|\nabla B(s\x)|\leq M h^{1/2}$. For the second bound we use the same estimate in \eqref{eq:deltaadivrep}.

Finally, by \eqref{eq:deltaarep2} we have
$$
\delta\A(\x) - \frac 13 \left[\nabla B(0)\cdot\x\right] (-x_2,x_1) =
\frac 12 \int_0^1 ds\, (1-s^2) \left[\left(\nabla B(s\x)-\nabla B(0)\right) \cdot \x\right] \, (-x_2,x_1)
$$
and, since $\nabla B(s\x)-\nabla B(0) 
= \int_0^s d\sigma\, \mathrm{Hess}\,B(\sigma \x)\, \x$,
$$
\delta\A(\x) - \frac 13 \left[\nabla B(0)\cdot\x\right] (-x_2,x_1) =
\frac12 \int_0^1 d\sigma \left(\frac23 -\sigma +\frac23\sigma^3\right)
\left( \x\cdot \mathrm{Hess}\,B(\sigma \x) \, \x\right) (-x_2,x_1) \,.
$$
In view of $\|\mathrm{Hess}\,B(\sigma \x)\| \leq M h$, we arrive at \eqref{eq:deltaataylor}.
\end{proof}

\begin{proof}[Proof of Proposition \ref{bulk}]

\emph{Step 1. Representation via the Helffer-Sj\"ostrand formula}

We may assume that $\x_0=0$. Moreover, since the diagonal of the 
integral kernels are gauge-invariant, we may assume that $\A$
 and $\A_0$ are given by \eqref{eq:a}. We write 
$H:= (-i\nabla - \A)^2$ and $H_0:= (-i\nabla - \A_0)^2$ in
$L^2(\R^2)$. 
Applying the resolvent identity twice, we find that for 
any $z\in\C\setminus[0,\infty)$ (the equality holds on functions with
compact support):  
$$
(H-z)^{-1} = (H_0-z)^{-1} - (H_0-z)^{-1} W (H_0-z)^{-1}
+(H_0-z)^{-1} W (H-z)^{-1} W (H_0-z)^{-1}
$$
where 
$$
W:= - (-i\nabla-\A_0) \cdot \delta\A - \delta\A\cdot (-i\nabla-\A_0) + |\delta\A|^2
$$
and $\delta\A := \A-\A_0$.

The above identity needs some justification since $W$ is unbounded,
and this is what we do next. Let us consider one piece building the product 
$W\,(H_0-z)^{-1}$, i.e. the operator \break 
$\delta\A\cdot
(-i\nabla -\A_0)\;(H_0-z)^{-1}$, and let
us apply it on a
compactly supported function $\psi$. Then we can write:
\begin{align}\label{oatreia1}
\delta\A\cdot
(-i\nabla -\A_0)\;(H_0-z)^{-1}\psi =&
e^{-{\frac{\delta \eta}{2r }|\cdot|}}\{e^{-{\frac{\delta \eta}{2r
    }|\cdot|}}\delta\A\}\nonumber \\
& \cdot \{e^{{\frac{\delta \eta}{r }|\cdot|}}(-i\nabla
-\A_0)\;(H_0-z)^{-1}e^{-{\frac{\delta
      \eta}{r }|\cdot|}}\}
\{e^{{\frac{\delta \eta}{r }|\cdot-\x_0|}}\psi\}.
\end{align}
Using the general estimates established in \eqref{hcadoua7} and
\eqref{hcapatra1} it follows that the resolvent sandwiched with
exponentials is bounded, and we also gain some exponential decay to the
left. Half of it is used to bound the linear growth in
$\delta\A$, while the other half can be used as a priori
decay in case we have other $W$'s and resolvents to bound. Moreover,
the norms  grow at most polynomially in $|\Re z|$ and in $1/|\Im z|$.

Hence by the Helf\-fer--Sj\"ostrand formula for any almost analytic
extension $f_{a}$ of $f$ satisfying the decay condition
\eqref{apatra2}, we can write the following operator identity which
holds at least on compactly supported functions:
\begin{align}\label{eq:hs}
f(H) - f(H_0) = & - \frac{1}{\pi}\int_{\mathcal{D}} \frac{\partial
  f_{a}}{\partial \overline{z}}(H_0-z)^{-1} W (H_0-z)^{-1} \,dxdy \\ 
& + \frac{1}{\pi}\int_{\mathcal{D}} \frac{\partial
  f_{a}}{\partial \overline{z}}(H_0-z)^{-1}W (H-z)^{-1} W(H_0-z)^{-1} \,dxdy \notag 
\end{align}
where $\mathcal D = \{z\in\mathbb{C}:\; |\Im z|<1\}$. Hence we need to prove estimates on the diagonal of the integral kernels of $(H_0-z)^{-1} W (H_0-z)^{-1}$ and $(H_0-z)^{-1}W (H-z)^{-1} W(H_0-z)^{-1}$. In order to do this, we need use the estimates from Lemma \ref{deltaa} on $\delta\A$ and its divergence. In the remainder of this proof we always assume that $z\in\mathcal{D}$.

\bigskip

\emph{Step 2. The term $(H_0-z)^{-1} W (H_0-z)^{-1}$.}
We write
\begin{align*}
(H_0-z)^{-1} W (H_0-z)^{-1}
= & -  \left( (-i\nabla-\A_0) (H_0-\overline z)^{-1} \right)^* \cdot \delta\A (H_0-z)^{-1} \\
& - (H_0-z)^{-1} \delta\A \cdot (-i\nabla-\A_0) (H_0- z)^{-1} \\
& + (H_0-z)^{-1} |\delta\A|^2 (H_0-z)^{-1} \,.
\end{align*}
Since the first term on the right side coincides with the adjoint of the second term after replacing $z$ by $\overline z$, it suffices to consider the second and the third term. 

We begin with the second term. Let $\Green_0(\cdot,\cdot,z)$ be the integral kernel of $(H_0-z)^{-1}$ and $\mathbf \Green_1(\cdot,\cdot,z)$ be the ($\C^2$-valued) integral kernel of $(-i\nabla-\A_0)(H_0-z)^{-1}$. Then the integral kernel of $(H_0-z)^{-1} \delta\A \cdot (-i\nabla-\A_0) (H_0- z)^{-1}$ is
\begin{equation}\label{eq:kernelmain}
\int_{\R^2} d\y \, \Green_0(\x,\y,z) \delta\A(\y) \cdot \mathbf \Green_1(\y,\x',z) \,.
\end{equation}
Since
$$
|\Green_0(\x,\x',z)| \leq C(z) (1+|\ln|\x-\x'||) e^{-|\x-\x'|/C(z)} \,, 
\ |\mathbf \Green_1(\x,\x',z)| \leq C(z) |\x-\x'|^{-1} e^{-|\x-\x'|/C(z)}
$$
by Propositions \ref{diamag} and \ref{kernelder} and since $|\delta\A(\x)| \leq M|\x|$ by Lemma \ref{deltaa} this integral converges absolutely for any $z\in\C\setminus\R$ and for any $\x$ and $\x'$. Moreover, one can show that the above integral depends continuously on $(\x,\x')$. 

Therefore we can and will set $\x=\x'=0$ in the following. In order to
estimate the integral \eqref{eq:kernelmain} we split it according
 to whether $|\y|\leq \ell$ or not. In the first case we note that
 $\Green_0(0,\y,z)$ depends on $|\y|$ only. 
This follows from Mehler's Formula (see \eqref{eq:2}).
Hence the tangential component of $\mathbf \Green_1(\y,0,z)$ is also a radial function times $(-y_2,y_1)$. This implies that for any fixed vector $\mathbf c\in\R^2$,
\begin{equation}\label{eq:kernelmain'}
\Green_0(0,\y,z) \ \left( \mathbf c \cdot \y \right)\ (-y_2,y_1)\cdot \mathbf \Green_1(\y,0,z)
\end{equation}
is odd with respect to $\y\mapsto-\y$, and therefore its integral over
$\{|\y|\leq \ell \}$ is zero. Hence by \eqref{eq:deltaataylor},
\begin{align*}
\left| \int_{|\y|\leq \ell} d\y\,  \Green_0(0,\y,z) \delta\A(\y)\cdot \mathbf \Green_1(\y,0,z) \right|
\leq \frac18 M h \int_{|\y|\leq \ell } d\y\,  \left| \Green_0(0,\y,z)\right| |\y|^3 \left| \mathbf \Green_1(\y,0,z) \right| \,.
\end{align*}
By Propositions \ref{diamag} and \ref{kernelder} there are $C$ and $\delta$ such that for all $z\in\mathcal D$ the latter integral is bounded by
$$
C \frac{r^{10}}{\eta^{2}} \int_{\R^2} d\y \, \left(1+|\ln|\y||\right) |\y|^2 e^{-\frac{\delta \eta}r |\y|} \,. 
$$
Recall that we are using the notations $r= \sqrt{\langle \Re z\rangle}$ and $\eta = |\Im z|$. One easily sees that for any $\epsilon>0$ there is a $C_\epsilon$ such that the latter is bounded for all $z\in\mathcal D$ by
\begin{equation}
 \label{eq:firstint}
C_\epsilon \frac{r^{10}}{\eta^{2}} \left( \frac{r}{\eta} \right)^{4+\epsilon} \,.
\end{equation}

In the region where $|\y|> \ell$, we use the first estimate in \eqref{eq:deltaaglob} to get
\begin{align*}
\left| \int_{|\y|> \ell } d\y\,  \Green_0(0,\y,z) \delta\A(\y)\cdot \mathbf \Green_1(\y,0,z) \right|
\leq M \int_{|\y|> \ell} d\x\,  \left| \Green_0(0,\y,z)\right| |\y| \left|\mathbf \Green_1(\y,0,z) \right| \,.
\end{align*}
By Propositions \ref{diamag} and \ref{kernelder} there are $C$ and $\delta$ such that for all $z$ with $|\Im z|\leq 1$ the latter integral is bounded by
$$
C \frac{r^{10}}{\eta^{2}} \int_{|\y|>\ell} d\y \, \left(1+|\ln|\y||\right) e^{-\frac{\delta \eta}r |\y|} \,.
$$
One easily deduces that there is a $C$ such that the latter is bounded for all $z\in\mathcal D$ by
$$
C \frac{r^{10}}{\eta^{2}} \left( \frac{r}{\eta} \right)^2 e^{-\frac{\delta \eta \ell}r} 
\, \ell^2 \,.
$$
(This can certainly be improved.) Using $t^N e^{-t}\leq (N/e)^N$, we conclude that for any $N\geq 0$ there is a $C_N$ such that for all $z\in\mathcal D$, the latter is bounded by
$$
C_N \frac{r^{10}}{\eta^{2}} \left( \frac{r}{\eta} \right)^{2+N} \, \ell^{2-N} \,.
$$
This, together with \eqref{eq:firstint}, proves that for any $N\geq 0$ there are $C_N$ and $L$ such that for all $z\in\mathcal D$,
$$
\left| \int_{\R^2} d\y \, \Green_0(0,\y,z) \delta\A(\y) \cdot \mathbf \Green_1(\y,0,z) \right|
\leq C_N \left( \frac{r}{\eta} \right)^L \left(h+ \ell^{-N} \right) \,.
$$

Next, we briefly turn to the term $(H_0-z)^{-1} |\delta\A|^2 (H_0-z)^{-1}$. As before one shows that this has a continuous integral kernel and that for any $N\geq 0$ there are $C_N$ and $L$ such that for all $z\in\mathcal D$
$$
\left| \int_{\R^2} d\y \, \Green_0(0,\y,z) |\delta\A(\y)|^2 \Green_0(\y,0,z) \right|
\leq C_N \left( \frac{r}{\eta} \right)^L \left(h+ \ell^{-N} \right) \,.
$$
This proof uses the first inequalities in \eqref{eq:deltaaglob} and \eqref{eq:deltaaloc}. The latter leads immediately to a power $(h^{1/2})^2=h$ and therefore one does not need a symmetry argument in this case.

Since for any given $L$, the almost analytic extension $f_a$ in
\eqref{eq:hs} can be chosen to satisfy for some constant $C_L$, 
$$
\left|\frac{\partial f_a}{\partial\overline z} (z) \right| \leq C_L \left(\frac{\eta}{r^2}\right)^L
$$
for all $z\in\mathcal D$, the above estimates prove that the integral kernel of
$$
\int_{\mathcal{D}} \frac{\partial f_{a}}{\partial \overline{z}}(H_0-z)^{-1} W (H_0-z)^{-1}(0,0) \,dxdy
$$
at $(0,0)$ exists and is bounded in absolute value by $C_N \left(h+ \ell^{-N} \right)$ for any $N$. 
\bigskip

\emph{Step 3. The term $(H_0-z)^{-1}W (H-z)^{-1} W(H_0-z)^{-1}$.}

Using that $(-i\nabla-\A_0) \cdot \delta\A = \delta\A \cdot (-i\nabla-\A_0) - i \Div\delta\A$ we have two representations
$$
W= - 2(-i\nabla-\A_0) \cdot \delta\A + V^*
$$
and
$$
W= - 2\delta\A \cdot (-i\nabla-\A_0) + V
$$
where $V:= |\delta\A|^2 + i \Div\,\delta\A$. This allows us to write
\begin{align}\label{eq:second}
 (H_0-z)^{-1} & W (H-z)^{-1} W(H_0-z)^{-1} \\
= & 4 \left( (-i\nabla-\A_0) (H_0-\overline z)^{-1} \right)^* \cdot \delta\A (H-z)^{-1} \delta\A \cdot (-i\nabla-\A_0) (H_0-z)^{-1} \notag \\
& - 2 (H_0-z)^{-1} V^* (H-z)^{-1} \delta\A \cdot (-i\nabla-\A_0) (H_0-z)^{-1} \notag\\
& - 2 \left( (-i\nabla-\A_0) (H_0-\overline z)^{-1} \right)^* \cdot \delta\A (H-z)^{-1} V (H_0-z)^{-1} \notag\\
& + (H_0-z)^{-1}V (H-z)^{-1} V(H_0-z)^{-1} \notag
\end{align}
We begin by considering the first summand. If $\Green(\cdot,\cdot,z)$ denotes the integral kernel of $(H- z)^{-1}$, then the integral kernel of $\left( (-i\nabla-\A_0) (H_0-\overline z)^{-1} \right)^* \cdot \delta\A (H-z)^{-1} \delta\A \cdot (-i\nabla-\A_0) (H_0-z)^{-1}$ is given by
$$
\int_{\R^2} d\y \int_{\R^2} d\y' \, \overline{ \mathbf \Green_1(\y,\x,\overline z)} \cdot \delta\A(\y) \, \Green(\y,\y',z) \, \delta\A(\y') \cdot \mathbf \Green_1(\y',\x',z)
$$
As before, it follows from Lemmas \ref{diamag} and \ref{kernelder} and the first estimate in \eqref{eq:deltaaglob} that this integral converges absolutely and depends continuously on $(\x,\x')$.

Moreover, the same argument as in Step 2 shows that the above integral, evaluated at $\x=\x'=0$ coincides with
$$
\int_{|\y|\leq \ell} d\y \int_{|\y'|\leq \ell} d\y' \, \overline{ \mathbf \Green_1(\y,0,\overline z)} \cdot \delta\A(\y) \, \Green(\y,\y',z) \, \delta\A(\y') \cdot \mathbf \Green_1(\y',0,z)
$$
up to an error which for any $N$ can be bounded by $C_N \left( \frac{r}{\eta} \right)^L \ell^{-N}$ for some $C_N$ and some $L=L(N)$. According to the first estimate in \eqref{eq:deltaaloc} this integral is bounded by
$$
\frac19 M^2 h \int_{|\y|\leq \ell} d\y \int_{|\y'|\leq \ell} d\y' \,
|\mathbf \Green_1(\y,0,\overline z)|\, |\y|^2\, |\Green(\y,\y',z)|\, |\y'|^2 \, |\mathbf \Green_1(\y',0,z)| \,.
$$
According to Propositions \ref{diamag} and \ref{kernelder} this double integral is bounded by
$$
C \frac{r^{16}}{\eta^3} \int_{\R^2} d\y \int_{\R^2} d\y' \, e^{-\frac{\delta\eta}{r}|\y|} |\y| \left(1+|\ln|\y-\y'|\,|\right) e^{-\frac{\delta\eta}{r}|\y-\y'|} |\y'| e^{-\frac{\delta\eta}{r}|\y'|} \,.
$$
One easily checks that for any $\epsilon>0$ this integral is bounded by $C_\epsilon \left(\frac{r}{\delta\eta}\right)^{4+\epsilon}$.

The other summands on the right side of \eqref{eq:second} are treated
similarly. For the integral kernel of $(H_0-z)^{-1}V (H-z)^{-1}
V(H_0-z)^{-1}$ one again first localizes to $|\y'|\leq \ell $ and
$|\y|\leq \ell $ and then estimates $V(\y)$ by \eqref{eq:deltaaloc}. The two factors of $V$ lead to a factor of $(h^{1/2})^2 = h$. In the second and the third summand in \eqref{eq:second} one can apply the Schwarz inequality and use the estimates for the first and the fourth summand.
\end{proof}


\section{Local asymptotics for operators on the half-plane}\label{halfspace}

In Corollary \ref{bulkcor} we studied $\Tr g f((-i\nabla -\A)^2 )$ for the operator $(-i\nabla -\A)^2$ defined on the whole space. Our next goal is to study similar traces for operators defined on a domain with boundary and for functions $g$ supported close to the boundary. Once again, the idea is that if $B$ varies slowly (measured in terms of a parameter $h$) on the support of the function $g$ (measured in terms of a parameter $\ell$), then an operator with constant magnetic field provides a good approximation. In contrast to Corollary \ref{bulkcor}, however, the presence of a boundary leads to the appearance of a second term involving an integral over the boundary. As a first step towards proving this, we consider in this section the case of the half-plane $\R^2_+=\{\x\in\R^2:\ x_2>0\}$.

\begin{proposition}\label{bdry}
Given $f\in\mathcal S(\R)$, $M>0$ and $N\geq 0$, there is a constant $C>0$ such that for any $B \in C^1(\R^2_+)$ satisfying 
$\|B\|_\infty+\|\nabla B\|_\infty \leq M$, for any $\x_0\in \R\times\{0\}$, $h\in(0,1]$ and $\ell\in[1,\infty)$ such that
\begin{equation}\label{eq:bloc}
\sup_{\x\in\R^2_+\,,\, |\x-\x_0|\leq \ell} |\nabla B(\x)| \leq M h^{1/2} \,,
\end{equation}
for any $\A\in C^2(\R^2_+,\R^2)$ and $\A_0\in C^2(\R^2_+,\R^2)$ satisfying $\curl \A=B$ and $\curl \A_0\equiv B(\x_0)$ and for any $g\in L^\infty(\{|\cdot-\x_0|\leq \ell /2\})$ one has
\begin{equation}\label{eq:bdry}
\left| \Tr g f((-i\nabla -\A)^2 ) - \Tr g f((-i\nabla -\A_0)^2 ) \right| 
\leq C \ell^2 \|g\|_\infty (h^{1/2} \ell^2 + h \ell^4 + \ell^{-N}) \,.
\end{equation}
Here $(-i\nabla -\A)^2$ and $(-i\nabla -\A_0)^2$ are the Dirichlet realizations of the corresponding operators in $L^2(\R^2_+)$.
\end{proposition}

We emphasize that, in contrast to Proposition \ref{bulk}, we only impose conditions on $B$ and its first derivative. Therefore the result is not as precise as those from Proposition~\ref{bulk} and Corollary~\ref{bulkcor}. This, however, will be sufficient for the application we have in mind. In our proof of Theorem~\ref{theorem2} we will choose $\ell=h^{-\epsilon}$ for some $\epsilon>0$. If $\epsilon<1/6$ and $N$ is chosen large enough, then the error will be $o(\ell)=o(h^{-\epsilon})$, whereas we shall compute $\Tr g f((-i\nabla -\A_0)^2 )$ up to order $\mathcal O(\ell)$.

Before turning to the proof of Proposition \ref{bdry} we deduce the following consequence.

\begin{corollary}\label{bdrycor}
Given $f\in\mathcal S(\R)$ and $M>0$, there is a constant $C>0$ with the following property: 
Let $B \in C^1(\R^2_+)$ satisfy $\|B\|_\infty+\|\nabla B\|_\infty \leq M$ and let $\x_0\in \R\times\{0\}$, $h\in(0,1]$ and $\ell\in[1,\infty)$ satisfy
\begin{equation}\label{eq:bloccor}
\sup_{\x\in\R^2_+\,,\, |\x-\x_0|\leq \ell} |\nabla B(\x)| \leq M h^{1/2}
\quad\text{and}\quad
\inf_{\x\in\R^2_+\,,\, |\x-\x_0|\leq \ell} B(\x) \geq M^{-1}
\,.
\end{equation}
Finally, assume that $g\in C^1(\{|\cdot-\x_0|\leq \ell /2\})$ satisfies
\begin{equation}\label{eq:gass}
\sup_{\x\in \R^2_+,\, |\x-\x_0|\leq \ell/2} |g(\x)| + \ell^{-1} \sup_{\x\in \R^2_+,\, |\x-\x_0|\leq \ell/2} |\nabla g(\x)| \leq M \,.
\end{equation}
Then for any $\A\in C^2(\R^2_+,\R^2)$ with $\curl \A=B$ one has
\begin{align}\label{eq:bdrycor}
& \left| \Tr g f((-i\nabla -\A)^2 ) - \frac1{2\pi} \sum_{k=1}^\infty \int_{\R^2} g(\x) b_k(B(\x),f) B(\x) \,d\x 
\right. \notag\\
& \qquad\qquad\qquad\qquad\ \ \left. 
- \frac1{2\pi} \sum_{k=1}^\infty \int_{\R} g(x_1,0) s_k(B(x_1,0),f) \sqrt{B(x_1,0)} \,dx_1 \right| \notag\\
& \qquad\qquad\leq C \ell^2 \left(h^{1/2}\ell^2 + h\ell^4 + \ell^{-2} \right) \,.
\end{align}
Here $(-i\nabla -\A)^2$ is the Dirichlet realization of the corresponding operator in $L^2(\R^2_+)$.
\end{corollary}

\begin{remark}\label{general}
There is nothing special about the $1/2$ in the support condition on $g$ in Proposition~\ref{bdry} and Corollary \ref{bdrycor}. Indeed, the proof below works equally well when the support of $g$ is contained in $\{|\cdot-\x_0|\leq \theta \ell \}$ for a fixed ($\ell$-independent) $\theta\in(0,1)$. Of course, then the constant will also depend on $\theta$. Similarly, there is nothing special about the lower bound $1$ on $\ell$. Again the proof works equally well when $\ell\in[\ell_0,\infty)$ for a fixed ($h$-independent) $\ell_0>0$. The constant will also depend on $\ell_0$.
\end{remark}

\begin{proof}[Proof of Corollary \ref{bdrycor}]
Our goal is to replace the term $\Tr g f((-i\nabla -\A_0)^2 )$ in Proposition \ref{bdry} by the two integrals on the left side of \eqref{eq:bdrycor} and to control the error. To do this, we first note that the operator $(-i\nabla -\A_0)^2$ can be diagonalized explicitly. By gauge invariance we may assume that $\A_0(\x)=b\, (-x_2,0)$ with $b=B(\x_0)$. Applying a Fourier transform with respect to the $x_1$ variable, we arrive at the operator $-\frac{d^2}{dx_2^2} + (p_1 + bx_2)^2$. Via the change of variables $t =b^{\frac 12} x_2$ this operator is unitarily equivalent to $ b\left(-\frac{d^2}{dt^2} + (b^{-\frac 12} p_1 + t)^2\right)$. This is the operator \eqref{modello} with $\xi = b^{-\frac12} p_1$. We conclude that the eigenvalues and eigenfunctions of the operator $-\frac{d^2}{dx_2^2} + (p_1 + bx_2)^2$ are given by $b e_j(b^{-\frac 12} p_1)$ and $b^\frac 14 \psi_j(b^{\frac 12} x_2,b^{-\frac 12} p_1 )$. Hence
\begin{align}\label{asasea6}
& \Tr g f((-i\nabla -\A_0)^2 ) = \frac{b}{2\pi} \sum_{k=1}^\infty \int_{\R^2} g(\x) b_k(b,f) \,d\x \\
& + \frac{1}{2\pi}\sum_{k=1}^\infty \int_{\R}dx_1\int_{0}^\infty dx_2 \, g(\x)
\left\{ \int_{\R} b^{1/2} f(b e_k(b^{-\frac 12} p_1)) |\psi_k(b^\frac12 x_2,b^{-\frac12} p_1)|^2 dp_1 - bf(b(2k-1))\right\} \,. \nonumber
\end{align}
Since the function $ \lambda \mapsto \sum_{k} f(\lambda(2k-1))$ is in $C^1(0,\infty)$ and since $\nabla B$ is controlled on the support of $g$ by $M h^{1/2}$, we have the estimate
\begin{equation}\label{asasea7}
 \left \vert \frac{b}{2\pi} \sum_{k=1}^\infty \int_{\R^2} g(\x) b_k(b,f) \,d\x 
 - \frac1{2\pi} \sum_{k=1}^\infty \int_{\R^2} g(\x) b_k(B(\x),f) B(\x) \,d\x \right\vert 
 \leq C\;h^{1/2} \ell^2 \,.
\end{equation}
Here we also took into account that $|\supp g |\leq (\pi/8) \ell^2$.

The second term on the right side of \eqref{asasea6} needs more care. After rescaling $x_2$ and $p_1$ it equals
$$
\frac{1}{2\pi}\sum_{k=1}^\infty \int_{\R}dx_1\int_{0}^\infty dt\, g(x_1,t b^{-\frac12}) b^{1/2}
\left\{ \int_{\R} f(b e_k(\xi)) |\psi_k(t,\xi)|^2 d\xi - f(b(2k-1))\right\}
$$
First, we replace $g(x_1, t b^{-\frac12})$ by $g(x_1,0)$. Because of the bound on the derivative of $g$ we have
$$
|g(x_1, t b^{-\frac12})-g(x_1,0)| \leq M \ell^{-1} t b^{-\frac12} \leq M^{3/2} \ell^{-1} t \,.
$$
In Lemma \ref{moment} we prove that
$$
\sum_{k=1}^\infty \int_{0}^\infty dt \, t \left| \int_{\R} f(b e_k(\xi)) |\psi_k(t,\xi)|^2 d\xi - f((2k-1)b)\right| <\infty \,.
$$
This allows us to replace $g(x_1, t b^{-\frac12})$ by $g(x_1,0)$ and then to carry out the $t$ integration. We obtain
\begin{align*}
& \left| \sum_{k=1}^\infty \int_{\R}dx_1\int_{0}^\infty dt\, g(x_1,t b^{-\frac12}) b^{1/2}
\left\{ \int_{\R} f(b e_k(\xi)) |\psi_k(t,\xi)|^2 d\xi - f((2k-1)b)\right\} \right. \notag \\
& \qquad \left. - \sum_{k=1}^\infty \int_{\R}dx_1\, g(x_1,0) b^{1/2} s_k(b,f) \right|
\leq C \,.
\end{align*}
Finally, we use the fact that the function $ \lambda \mapsto \sum_{k} s_k(\lambda,f)$ is in $C^1(0,\infty)$, see Lemma \ref{1.4bis}. This, together with our assumptions on $B$ and $g$, shows that
$$
\left| \sum_{k=1}^\infty \int_{\R}dx_1 g(x_1,0) b^{1/2} s_k(b,f) - \sum_{k=1}^\infty \int_{\R}dx_1 g(x_1,0) \sqrt{B(x_1,0)} s_k(B(x_1,0),f) \right| \leq C h^{1/2} \ell \,.
$$
(Note that in the last step we used $B(0)>0$.) This leads to the bound stated in the corollary.
\end{proof}



We now turn to the proof of Proposition \ref{bdry}. The following lemma is the analogue of Lemma \ref{deltaab}. We omit the proof, which follows the same lines.

\begin{lemma}\label{deltaab}
Assume that $B\in C^1(\R^2_+)$ satisfies $\| B \|_\infty+\|\nabla B\|_\infty \leq M$ and
$$
\sup_{\x\in\R^2_+\,,\, |\x|\leq \ell} |\nabla B(\x)| \leq M h^{1/2}
$$
for some positive $M, h, \ell$. We define
\begin{equation}\label{eq:ab}
\A(\x) := \int_0^1 dt \, B(x_1,tx_2) \, (-x_2,0)\,,
\qquad
\A_0(\x) := B(0)\, (-x_2,0) \,.
\end{equation}
and $\delta\A:=\A-\A_0$. Then for all $\x\in\R^2_+$
\begin{equation}\label{eq:deltaaglobb}
|\delta \A(\x)| \leq 2 M x_2
\qquad\text{and}\qquad
|\Div\, \delta \A(\x)| \leq M x_2 \,.
\end{equation}
Moreover, for all $\x\in\R^2_+$ with $|\x|\leq \ell$
\begin{equation}\label{eq:deltaalocb}
|\delta \A(\x)| \leq M h^{1/2} x_2 |\x|
\qquad\text{and}\qquad
|\Div\delta \A(\x)| \leq M h^{1/2} x_2 \,.
\end{equation}
\end{lemma}

\begin{proof}[Proof of Proposition \ref{bdry}]

\emph{Step 1. Representation via the Helffer-Sj\"ostrand formula}

We may assume that $\x_0=0$. Moreover, since the local traces are gauge-invariant, we may assume that $\A$ and $\A_0$ are given by \eqref{eq:ab}. We write $H:= (-i\nabla - \A)^2$ and $H_0:= (-i\nabla - \A_0)^2$ in $L^2(\R^2_+)$ with Dirichlet boundary conditions. With a justification as in the proof of Proposition~\ref{bulk} one has the resolvent identity in the form
$$
(H-z)^{-1} = (H_0-z)^{-1} - (H_0-z)^{-1} W (H-z)^{-1}
$$
where 
$$
W:= - (-i\nabla-\A_0) \cdot \delta\A - \delta\A\cdot (-i\nabla-\A_0) + |\delta\A|^2
$$
and $\delta\A := \A-\A_0$. Hence by the Helf\-fer--Sj\"ostrand formula for any almost analytic extension $f_{a}$ of $f$ we can write:
\begin{align*}
f(H) - f(H_0) = - \frac{1}{\pi}\int_{\mathcal{D}} \frac{\partial
  f_{a}}{\partial \overline{z}}(H_0-z)^{-1} W (H-z)^{-1} \,dxdy \,,
\end{align*}
where again $\mathcal D = \{z\in\mathbb{C}:\; |\Im z|<1\}$. In contrast to the proof of Proposition \ref{bulk}, however, this identity is not sufficient for our purposes, since the operator $(H_0-z)^{-1} W (H-z)^{-1}$, even when multiplied by cut-off functions from the left and from the right, is not trace class. Our way out is to note that the above integral does not change if we subtract an analytic (operator-valued) function to the integrand. We subtract $(H_0+1)^{-1} W (H+1)^{-1}$ and compute
\begin{align*}
& (H_0-z)^{-1} W (H-z)^{-1} - (H_0+1)^{-1} W (H+1)^{-1} \\
&\quad = (z+1) \left( (H_0-z)^{-1} (H_0+1)^{-1} W (H-z)^{-1} + (H_0+1)^{-1} W (H+1)^{-1}(H-z)^{-1} \right) \,.
\end{align*}
Thus
\begin{align*}
f(H) - f(H_0) = & - \frac{1}{\pi} \int_{\mathcal{D}} dxdy\, \frac{\partial f_{a}}{\partial \overline{z}} (z+1) \notag \\
 & \times \left( (H_0-z)^{-1} (H_0+1)^{-1} W (H-z)^{-1} + (H_0+1)^{-1} W (H+1)^{-1}(H-z)^{-1} \right) \,.
\end{align*}

Next, we denote by $\eta$ the characteristic function of the disk $\{ |\x|\leq \ell\}$. Then, by cyclicity of the trace and the fact that $\eta g = g$, we have
\begin{align*}
& \Tr g \left( f(H) - f(H_0) \right) = - \frac{1}{\pi} \int_{\mathcal{D}} dxdy\, \frac{\partial f_{a}}{\partial \overline{z}} (z+1) \notag \\
 & \qquad\times \Tr \left( \eta (H_0-z)^{-1} (H_0+1)^{-1} W (H-z)^{-1} g + g (H_0+1)^{-1} W (H+1)^{-1}(H-z)^{-1} \eta \right) \,.
\end{align*}
In the remainder of this proof we will show that the operators $\eta (H_0-z)^{-1} (H_0+1)^{-1} W (H-z)^{-1} g$ and $g (H_0+1)^{-1} W (H+1)^{-1}(H-z)^{-1} \eta$ are trace class and we derive bounds on their trace norms. 

We begin by bounding
$$
\| \eta (H_0-z)^{-1} (H_0+1)^{-1} W (H-z)^{-1} g \|_{B_1(L^2(\R^2_+))} \leq 
\| \eta (H_0-z)^{-1} (H_0+1)^{-1} \|_{B_1(L^2(\R^2_+))} \| W (H-z)^{-1} g \|
$$
and
$$
\| g (H_0+1)^{-1} W (H+1)^{-1}(H-z)^{-1} \eta \|_{B_1(L^2(\R^2_+))} \leq 
\| (H+1)^{-1}(H-z)^{-1} \eta \|_{B_1(L^2(\R^2_+))} \| g (H_0+1)^{-1} W \| \,.
$$
The trace class factors on the right side are controlled by
\begin{equation}
\label{eq:tracehs}
\| \eta (H_0+1)^{-2} \|_{B_1(L^2(\R^2_+))} \leq C \ell^2
\quad\text{and}\quad
\| (H+1)^{-2} \eta \|_{B_1(L^2(\R^2_+))} \leq C \ell^2 \,.
\end{equation}
These bounds follows from Proposition \ref{prop15} together with the fact that one needs $\mathcal O(\ell^2)$ squares of unit diameter to cover the support of $\eta$. In addition, we use the fact that $(H_0+1)(H_0-z)^{-1} = 1+ (z+1)(H_0-z)^{-1}$, which implies
\begin{equation}
\label{eq:z-1}
\| (H_0+1)(H_0-z)^{-1} \| \leq 1+|z+1| |\Im z|^{-1} \,.
\end{equation}
A similar bound holds for $H$.

To summarize, we have shown that
\begin{align}
\label{eq:hsbound}
\left| \Tr g \left( f(H) - f(H_0) \right) \right| \leq C \ell^2 \int_{\mathcal{D}} dxdy\, & \left| \frac{\partial f_{a}}{\partial \overline{z}}\right| |z+1| \left( 1+ \frac{|z+1|}{|\Im z|} \right) \notag \\
& \times \left( \| W (H-z)^{-1} g \| + \| g (H_0+1)^{-1} W \| \right) \,,
\end{align}

We now decompose $W$ into a part close to the origin and a part far from the origin. More precisely, we write
$$
W= \eta W + (1-\eta)W = W\eta + W(1-\eta) \,.
$$
We employ the two different forms for the two different operators we have to estimate. In the following two steps of the proof we will derive the bounds
\begin{equation}
\label{eq:w0}
\| \eta W (H-z)^{-1} g \| + \| g (H_0+1)^{-1} W \eta \| \leq C \|g\|_\infty \frac{r^2}{\eta} \left( h^{1/2} \ell^2 + h \ell^4 \right)
\end{equation}
and
\begin{equation}
\label{eq:w1}
\| (1-\eta) W (H-z)^{-1} g \| + \| g (H_0+1)^{-1} W (1-\eta) \| \leq C_N \| g \|_\infty \, \frac{r^{6+N}}{\eta^{5+N}} \, \ell^{-N}
\end{equation}
for $z\in\mathcal D$. We recall the notations $r=\sqrt{\langle \Re z\rangle}$ and $\eta=|\Im z|$. The constants $C$ and $C_N$ in \eqref{eq:w0} and \eqref{eq:w1} depend only on $M$ and the latter one also on $N\geq 1$. These bounds, when inserted into \eqref{eq:hsbound}, are integrable with respect to $z$ (provided the decay order $N$ in \eqref{apatra2} of the almost analytic extension is large enough) and lead to the behavior claimed in the proposition.

\bigskip

\emph{Step 2. Bounds on $\eta W$ and $W\eta$.}
The idea behind our estimate in this case is that the coefficients in the operators $\eta W$ and $W\eta$ are non-zero only close to the origin, and even there, they are small, as quantified by Lemma \ref{deltaab}. We shall show that
\begin{equation}
\label{eq:w0proof}
\| (H_0+1)^{-1/2} W \eta \|_\infty \leq C \left( h^{1/2} \ell^2 + h \ell^4 \right)
\quad\text{and}\quad
\| \eta W (H+1)^{-1/2} \|_\infty \leq C \left( h^{1/2} \ell^2 + h \ell^4 \right)
\end{equation}
with a constant $C$ depending only on $M$. This, together with \eqref{eq:z-1}, implies \eqref{eq:w0}.

To prove \eqref{eq:w0proof}, we compute
$$
W \, \eta = - 2 (-i\partial_1 - A_{0,1})\,\eta\, \delta A_{1} + \eta\, V^* \,,
$$
where again $V^*=(\delta A_1)^2 - i\partial_1\delta A_1$. Here $A_{0,1}$ denotes the first component of $\A_0$ and similarly for $\delta A_1$. The Schwarz inequality implies that
$$
\left| \left( \psi, W \eta \phi \right) \right| 
\leq 2 \| \eta\, \delta A_1 \|_\infty \| (-i\partial_1 - A_{0,1})\psi\| \|\phi\|
+ \| \eta V \|_\infty \|\psi\| \|\phi\| \,.
$$
Lemma \ref{deltaab} together with the properties of $\eta$ implies that $\| \eta\, \delta A_1 \|_\infty \leq 4 M h^{1/2} \ell^2$ and $\| \eta V \|_\infty \leq 16 M^2 h \ell^4+ 2M h^{1/2} \ell$. Thus
$$
\left| \left( \psi, W_0 \psi \right) \right| \leq C \left( h^{1/2} \ell^2 + h \ell^4 \right) (\psi, (H_0+1) \psi)^{1/2} \|\phi\|^2
$$
with $C= 8M + 16M^2 + 2M$. This proves the first inequality in \eqref{eq:w0proof}. The second one follows in the same way by using the representation
\begin{equation}
\label{eq:w0rep2}
\eta W = -2 \eta \, \delta A_1 \, (-i\partial_1 - A_{1}) - \eta V^* \,.
\end{equation}

\bigskip

\emph{Step 3. Bounds on $(1-\eta)W$ and $W(1-\eta)$.}
The idea behind our estimate in this case is that the coefficients in the operators $(1-\eta)W$ and $W(1-\eta)$ vanish on the support of $g$. This fact, combined with the exponential decay of the resolvents, leads to errors which decay faster than any polynomial in $\ell$. To make this precise we apply Proposition \ref{prop4} to the operators $(1-\eta)W$ and $W(1-\eta)$ and deduce that given $M$ there are constants $C$ and $\delta>0$ such that for any $z\in\mathcal D$ and for any characteristic functions $\chi$ and $\tilde \chi$ of disjoint sets of diameter $\leq 1$ and of distance $d\geq 1$ one has
$$
\| \chi (1-\eta) W (H-z)^{-1} \tilde\chi \| \leq C \frac{r^2}{\eta} e^{-\delta \eta d/r} \| \chi (1-\eta) \x^2 \|_\infty \,.
$$
The fact that $C$ only depends on $M$ follows from the fact that the left side of \eqref{hcadoua7} can be bounded in terms of $M$ and $\| \chi (1-\eta) \x^2 \|_\infty$. Indeed, by the argument following \eqref{hcadoua7} it suffices to bound the norm of
$e^{{\alpha \langle \cdot -\x_0\rangle }}\chi (1-\eta) W e^{-\alpha \langle \cdot -\x_0\rangle }(H_D-i)^{-1}$ in terms of these quantities. This follows easily from the representation (cf. \eqref{eq:w0rep2})
\begin{equation}
\label{eq:w1rep2}
(1-\eta) W = -2 (1-\eta)\, \delta A_1 \,(-i\partial_1 - A_{1}) - (1-\eta) V^* \,,
\end{equation}
together with the bounds from Lemma \ref{deltaab} on $\delta A_1$ and $V$.

We now choose the $\chi_j$'s to be characteristic functions of sets of diameter $\leq 1$ such that $\sum_j \chi_j$ is the characteristic function of the support of $1-\eta$. (For instance, let $\chi_j$ be the characteristic function of a square on the lattice $(2^{-1/2} \mathbb Z)^2$ times the characteristic function of the support of $1-\eta$.) Similarly, let the $\tilde\chi_k$'s be characteristic functions of sets of diameter $\leq 1$ such that $\sum_k \tilde\chi_k$ is the characteristic function of the support of $g$. Note that since $g$ is supported in $\{|\x|\leq \ell/2\}$ and $1-\eta$ is supported in $\{|\x|\geq \ell \}$, the distance $d_{j,k}$ between the supports of $\tilde\chi_k$ and $\chi_j$ is at least $\ell/2$. Thus the above bound allows one to deduce (for $\ell\geq 2$ at least; the remaining case is easily handled)
\begin{multline*}
\| (1-\eta) W (H-z)^{-1} g \|_\infty \leq \\
C \| g\|_\infty \frac{r^2}{\eta} \left( \sup_j \| \chi_j \x^2 \|_\infty \sum_k e^{-\delta \eta d_{j,k}/r}  \right)^{1/2} \left( \sup_k \sum_j e^{-\delta \eta d_{j,k}/r} \| \chi_j \x^2 \|_\infty \right)^{1/2} \,.
\end{multline*}
Using elementary estimates one can bound the product of the two square roots by a constant times $\ell (r/\eta)^3 e^{-\tilde\delta \eta\ell/r}$ for some $0<\tilde\delta<\delta$. Given $N\geq 0$ we use the bound $t^N e^{-t} \leq (N/e)^N$ to arrive at the first bound claimed in \eqref{eq:w1}. The proof of the second one is similar. This concludes the proof of Proposition \ref{bdry}.
\end{proof}


\section{Local asymptotics for operators on a perturbed half-plane}\label{phalfspace}

We now extend the results from the previous section to `perturbed half-planes'. By this we mean domains of the form
\begin{equation}
\label{marrtie1}
\mathcal P = \{ \y\in\R^2 :\ y_2 > q(y_1) \} \,.
\end{equation}
with a function $q\in C^2(\R)$. We note that the change of variables $\y\mapsto\x=(y_1,y_2-q(y_1))$ transforms $\mathcal P$ into $\R^2_+$. In this section, we will typically use a $\tilde{}\,\,$ to denote an object on $\R^2_+$ which corresponds to the un-tilded object on $\mathcal P$. What `corresponds' means depends on the object we are dealing with. For scalar functions $g$ on $\mathcal P$ we define functions $\tilde g$ on $\R^2_+$ by
$$
\tilde g(\x) = g(x_1,x_2+q(x_1)) \,.
$$
For vector fields $\A\in C^2(\mathcal P,\R^2)$ we define $\tilde\A$ on $\R^2_+$ by
$$
\tilde\A(\x) = \left(A_1(x_1,x_2+q(x_1)) + q'(x_1) A_2(x_1,x_2+q(x_1)), A_2(x_1,x_2+q(x_1))\right) \,.
$$
This definition is motivated by the fact that $\curl\tilde\A(\x) = \curl\A(x_1,x_2+q(x_1))$, that is, $\tilde\A$ generates the same magnetic field as $\A$ in the new coordinates. Finally, if $H=(-i\nabla_\y - \A)^2$ is a magnetic operator in $L^2(\mathcal P)$, with Dirichlet boundary conditions, say, then
$$
\tilde H = \left( -i\partial_{x_1} - \tilde A_1 - q' \left( -i \partial_{x_2} - \tilde A_2 \right) \right)^2 + \left( -i \partial_{x_2} - \tilde A_2 \right)^2
$$
denotes the corresponding operator on $L^2(\R^2_+)$. The operators $H$ and $\tilde H$ are unitarily equivalent via the above change of variables.

The next proposition, which is the main result of this section, quantifies in which sense $\tilde H$ is close to the simpler operator $(-i\nabla -\tilde\A)^2$ when the boundary is almost flat.

\begin{proposition}\label{bdry2}
Given $f\in\mathcal S(\R)$, $M>0$ and $N\geq 0$, there is a constant $C>0$ such that for any $q\in C^2(\R)$ satisfying $\|q'\|_\infty + \|q''\|_\infty \leq M$, for any $B \in C(\R^2_+)$ satisfying $\|B\|_\infty \leq M$, for any $\y_0\in \partial \mathcal P$, $h\in(0,1]$ and $\ell\in[1,M h^{-1/2}]$ such that $q'(y_{0,1})=0$ and
\begin{equation}\label{eq:qloc}
\sup_{\y\in\mathcal P\,,\, |\y-\y_0|\leq \ell} |q''(y_1)| \leq M h^{1/2} \,,
\end{equation}
for any $\A\in C^1(\mathcal P,\R^2)$ satisfying $\curl \A=B$ and for any $g\in L^\infty(\{|\cdot-\y_0|\leq \ell /2\})$ one has
\begin{equation}\label{eq:bdry2}
\left| \Tr_{L^2(\mathcal P)} g f((-i\nabla -\A)^2 ) - \Tr_{L^2(\R^2_+)} \tilde g f((-i\nabla - \tilde\A)^2 ) \right| 
\leq C \ell^2 \|g\|_\infty \left( h \ell + \ell^{-N}\right) \,.
\end{equation}
Here $(-i\nabla -\A)^2$ and $(-i\nabla -\tilde\A)^2$ are the Dirichlet realizations of the corresponding operators in $L^2(\mathcal P)$ and $L^2(\R^2_+)$, respectively.
\end{proposition}

This proposition, combined with our previous analysis about half-space operators leads to

\begin{corollary}\label{bdry2cor}
Given $f\in\mathcal S(\R)$ and $M>0$ there is a constant $C>0$ with the following property:
Let $q\in C^2(\R)$ satisfy $\|q'\|_\infty + \|q''\|_\infty \leq M$ and let $B \in C^1(\R^2_+)$ satisfy $\|B\|_\infty +\|\nabla B\|_\infty \leq M$. Moreover, let $\y_0\in \partial \mathcal P$, $h\in(0,1]$ and $\ell\in[1,C^{-1}h^{-1/2}]$ be such that $q'(y_{0,1})=0$,
\begin{equation}\label{eq:qloccor}
\sup_{\y\in\mathcal P\,,\, |\y-\y_0|\leq \ell} |q''(y_1)| \leq M h^{1/2} \,,
\end{equation}
and
\begin{equation}\label{eq:bloc2}
\sup_{\x\in\mathcal P\,,\, |\y-\y_0|\leq \ell} |\nabla B(\y)| \leq M h^{1/2}
\quad\text{and}\quad
\inf_{\x\in\mathcal P\,,\, |\y-\y_0|\leq \ell} B(\y) \geq M^{-1}
\,.
\end{equation}
Finally, assume that $g\in C^1(\{|\cdot-\y_0|\leq \ell /2\})$ satisfies
\begin{equation}\label{eq:gass2}
\sup_{\y\in \mathcal P,\, |\y-\y_0|\leq \ell/2} |g(\y)| + \ell^{-1} \sup_{\y\in \mathcal P,\, |\y-\y_0|\leq \ell/2} |\nabla g(\y)| \leq M \,.
\end{equation}
Then for any $\A\in C^2(\mathcal P,\R^2)$ with $\curl \A=B$ one has
\begin{align}\label{eq:bdry2cor}
& \left| \Tr g f((-i\nabla -\A)^2 ) - \frac1{2\pi} \sum_{k=1}^\infty \int_{\mathcal P} g(\y) b_k(B(\y),f) B(\y) \,d\y \right. \notag\\
& \qquad\qquad\qquad\left. - \frac1{2\pi} \sum_{k=1}^\infty \int_{\partial\mathcal P} g(\y) s_k(B(\y),f) \sqrt{B(\y)} \,d\sigma(\y) \right| \notag\\
& \qquad\qquad\qquad\leq C \ell^2 \left(h^{1/2}\ell^2 + h\ell^4 + \ell^{-2} \right) \,.
\end{align}
Here $(-i\nabla -\A)^2$ is the Dirichlet realization of the corresponding operator in $L^2(\mathcal P)$.
\end{corollary}

\begin{remark}\label{qassrem}
We emphasize that assumption \eqref{eq:qloccor} is weaker than what we need later in our application, where we have a similar bound but with the right hand side replaced by $Mh$. Imposing this stronger assumption, however, does not lead to an improvement of \eqref{eq:bdry2cor}. The reason is that the error in \eqref{eq:bdry2cor} comes from Corollary \ref{bdrycor} and the error coming from straightening the boundary in Proposition \ref{bdry2} is of lower order.
\end{remark}

\begin{proof}[Proof of Corollary \ref{bdry2cor}]
After a translation, if necessary, we may assume that $\y_0=0$. In particular, we have $q(0)=q'(0)=0$. According to Proposition \ref{bdry2} we can replace, up to the accuracy we are interested in, the trace $\Tr g f((-i\nabla -\A)^2 )$ by $\Tr_{L^2(\R^2_+)} \tilde g f((-i\nabla - \tilde\A)^2 )$. We shall show that the latter trace can be computed using Corollary \ref{bdrycor}. This will prove the corollary, since
$$
\frac1{2\pi} \sum_{k=1}^\infty \int_{\R^2} \tilde g(\x) b_k(\tilde B(\x),f) \tilde B(\x) \,d\x 
= \frac1{2\pi} \sum_{k=1}^\infty \int_{\mathcal P} g(\y) b_k(B(\y),f) B(\y) \,d\y
$$
and
\begin{align*}
\left| \frac1{2\pi} \sum_{k=1}^\infty \int_{\R} g(x_1,0) s_k(B(x_1,0),f) \sqrt{B(x_1,0)} \,dx_1
- \frac1{2\pi} \sum_{k=1}^\infty \int_{\partial\mathcal P} g(\y) s_k(B(\y),f) \sqrt{B(\y)} \,d\sigma(\y) \right| \leq C h \ell^3 \,.
\end{align*}
The last bound follows from the fact that $1\leq \sqrt{1+q'^2}\leq 1+ C h\ell^2$ on the support of $g$. This bound follows from \eqref{eq:qloc} in Lemma~\ref{q} below.

We now verify the assumptions of Corollary \ref{bdrycor}. As observed before, $\tilde A$ generates the magnetic field $\tilde B(x)=B(x_1,x_2+q(x_1))$. Thus the assumption $\|B\|_\infty +\|\nabla B\|_\infty \leq M$, together with $\|q'\|_\infty \leq M$, implies the corresponding assumption for $\tilde B$. In order to prove the local assumption on $B$ and $g$ we need the following two inclusions,
\begin{equation}
\label{eq:sets}
\{ \x\in\R^2_+ :\ |\x| \leq A^{-1} L \} \subset \{ \x\in\R_+^2:\ x_1^2 + (x_2 + q(x_1))^2 \leq L^2 \} 
\subset \{ \x\in\R^2_+ :\ |\x| \leq A L \}
\end{equation}
valid for all $L\in (0,\ell]$ with $A=1+Ch^{1/2}\ell$, where $C$ only depends on $M$ and where we assume that $h^{1/2}\ell\leq 1$. Accepting \eqref{eq:sets} for the moment, we deduce that the bounds on $B$ and $g$ on the sets $\{\y \in\mathcal P:\ |\y|\leq \ell\}$ and $\{\y \in\mathcal P:\ |\y|\leq \ell/2\}$, respectively, will imply corresponding bounds for $\tilde B$ and $\tilde g$ on the sets $\{\x\in \R^2_+:\ |\x|\leq A^{-1} \ell\}$ and $\{\x \in\R^2_+:\ |\x|\leq A \ell/2\}$. If $h^{1/2}\ell\leq (3C)^{-1}$, where $C$ is the constant entering in the definition of $A$, then the ratio of the radii of the two disks is $A^2/2\leq (4/3)^2/2=\theta<1$. The assertion now follows from Corollary \ref{bdrycor} taking Remark \ref{general} into account.

We now turn to the proof of \eqref{eq:sets}. Assume first that $|\x|\leq cL$ for some $c\leq 1$ and $L\leq \ell$. Then, in particular, $|x_1|\leq cL$ and $|q(x_1)|\leq Q cL$ where $Q=(M/2)(1+M/4) h^{1/2}\ell$ from \eqref{eq:qloc}. This implies
$$
(x_2+q(x_1))^2 \leq x_2^2 + 2 x_2 |q(x_1)| + q(x_1)^2 \leq x_2^2 + 2 Q c^2 L^2 + Q^2 c^2 L^2
$$
and therefore $x_1^2 + (x_2+q(x_1))^2 \leq c^2 L^2 (1+ 2Q + Q^2)$. This implies the first part of the assertion with any $A\geq 1+Q$. Conversely, assume that $x_1^2 + (x_2+q(x_1))^2 \leq L^2$ for some $L \leq \ell$. Then, as before, $|q(x_1)|\leq QL$. This implies $x_2 \leq |x_2+q(x_1)|+|q(x_1)|\leq L(1+Q)$ and therefore
$$
x_2^2 \leq (x_2+q(x_1))^2 + 2 x_2 |q(x_1)| - q(x_1)^2 \leq (x_2+q(x_1))^2 + 2 L^2 Q(1+Q)
$$
Hence $|\x|^2 \leq L^2 (1+ 2Q(1+Q^2))$. This implies the second part of the assertion with any $A\geq 1+ 2Q(1+Q^2)$. Taking into account that $h^{1/2}\ell\leq 1$, we obtain what we claimed in \eqref{eq:sets}. This completes the proof of Corollary \ref{bdry2cor}.
\end{proof}



Before beginning with the proof of Proposition \ref{bdry2} we state and prove

\begin{lemma}\label{q}
Assume that $q\in C^2(\R)$ satisfies $q(0)=q'(0)=0$, $\|q''\|_\infty \leq M$ and
\begin{equation*}
\sup_{\y\in\mathcal P\,,\, |\y|\leq \ell} |q''(y_1)| \leq M h^{1/2} \,,
\end{equation*}
for some positive $M, h,\ell$. Then $|q'(y_1)|\leq M |y_1|$ for all $y_1\in\R$. Moreover, for all $y_1$ with $|y_1|\leq \ell$
\begin{equation}
\label{eq:qlocbound}
|q'(y_1)|\leq M(1+M/4) h^{1/2} |y_1|
\quad\text{and}\quad
|q(y_1)|  \leq (M/2)(1+M/4) h^{1/2} y_1^2 \,.
\end{equation}
\end{lemma}

\begin{proof}
The bound $|q'(y_1)|\leq M |y_1|$ simply follows by integrating the global bound $\|q''\|_\infty\leq M$. Similarly, the second bound in \eqref{eq:qlocbound} follows by integrating the first one. To prove the first inequality in \eqref{eq:qlocbound}, we observe that by assumption we have $q'(0)=0$ and $|q''(y_1)|\leq M h^{1/2}$ for any $y_1$ for which there is an $x_2>0$ such that $y_1^2 + (x_2+q(y_1))^2 \leq \ell^2$. Easy geometric considerations show that such $x_2$ exists at least when $|y_1|\leq \ell (1+(M/2) h^{1/2} )^{-1/2}$. (Indeed, the domain $\mathcal P\cap\{|\y|\leq \ell\}$ lies between the two parabolae $x_2=\pm (M/2) h^{1/2} y_1^2$.) Thus for all such $y_1$ we have $|q'(y_1)| \leq M h^{1/2} |y_1|$. On the other hand, for $\ell (1+(M/2) h^{1/2})^{-1/2} \leq |y_1|\leq \ell$ we use in addition the global bound $|q''|\leq M$ and find
\begin{align*}
|q'(y_1)| & \leq M h^{1/2} (1+(M/2) h^{1/2} )^{-1/2} \ell + M (|y_1| - \ell (1+(M/2) h^{1/2})^{-1/2} ) \\
& \leq M |y_1| \left( h^{1/2} + 1-(1+(M/2) h^{1/2})^{-1/2} \right) \,.
\end{align*}
The first bound in \eqref{eq:qlocbound} now follows from the elementary inequality $1-(1+x)^{-1/2}\leq x/2$ for $x\geq 0$.
\end{proof}

We now give the

\begin{proof}[Proof of Proposition \ref{bdry2}]
The proof is similar to that of Proposition \ref{bdry} and we only sketch the major differences. As in the proof of Corollary \ref{bdry2cor} we may assume that $\y_0=0$, so that $q(0)=q'(0)=0$. We have
$$
\Tr_{L^2(\mathcal P)} g f((-i\nabla -\A)^2 ) = \Tr_{L^2(\R^2_+)} \tilde g f(\tilde H) \,,
$$
where $\tilde H$ was introduced before the proposition. We write
$$
\tilde H = H_0 + \tilde W
$$
with $H_0 = (-i\nabla - \tilde\A)^2$ and
$$
\tilde W = - (-i\partial_1 -\tilde A_1) q' (-i\partial_2 -\tilde A_2) - (-i\partial_2 -\tilde A_2)q'(-i\partial_1 -\tilde A_1) + (-i\partial_2 -\tilde A_2)(q')^2 (-i\partial_2 -\tilde A_2) \,.
$$
(Strictly speaking, the notation $\tilde H_0$ instead of $H_0$ would be more consistent, but we try to keep the notation simple at this point.) Similarly as in the proof of Proposition \ref{bdry} we denote by $\eta$ the characteristic function of the disk $\{ |\y|\leq \ell\}$. Using the Helffer--Sj\"ostrand formula we find as before that
\begin{align*}
& \Tr \tilde g \left( f(\tilde H) - f(H_0) \right) = - \frac{1}{\pi}\int_{\mathcal{D}} dxdy\, \frac{\partial f_{a}}{\partial \overline{z}} (z+1) \\
& \qquad \times \Tr \left( \tilde\eta (H_0-z)^{-1} (H_0+1)^{-1} \tilde W (\tilde H-z)^{-1} \tilde g + \tilde g (H_0+1)^{-1} \tilde W (\tilde H+1)^{-1}(\tilde H-z)^{-1} \tilde\eta \right) \,,
\end{align*}
and we need to show that the operators $\tilde\eta (H_0-z)^{-1} (H_0+1)^{-1} \tilde W (\tilde H-z)^{-1} \tilde g$ and $\tilde g (H_0+1)^{-1} \tilde W (\tilde H+1)^{-1}(\tilde H-z)^{-1} \tilde\eta$ are trace class with a suitable bound on their trace norms. 

The factors $\tilde\eta (H_0+1)^{-2}$ and $(\tilde H+1)^{-2} \tilde\eta$ are trace class with norms as in \eqref{eq:tracehs}. Here we also use the fact that the support of $\tilde\eta$ is contained in $\{|\x| \leq A\ell^2\}$ with a constant $A$ depending only on $M$, which follows as in the proof of \eqref{eq:sets}. (At this point we use the assumption that $h^{1/2}\ell \leq M$.)

Thus it remains to control the operator norm of $\tilde W (\tilde H-z)^{-1} \tilde g$ and $\tilde g (H_0+1)^{-1} \tilde W$. We begin with the second operator. Again, we decompose $\tilde W=\tilde W\tilde\eta+\tilde W (1-\tilde\eta)$.

We now explain how to bound the operator $\tilde g (H_0+1)^{-1} \tilde W\tilde\eta$. We commute all derivatives in $\tilde W$ to the left and find
$$
\tilde W = - (-i\partial_1 -\tilde A_1)(-i\partial_2 -\tilde A_2) q' - (-i\partial_2 -\tilde A_2)(-i\partial_1 -\tilde A_1)q' + (-i\partial_2 -\tilde A_2)^2 (q')^2   -i(-i\partial_2 -\tilde A_2) q'' \,.
$$
Next, we commute one derivative through $(H_0+1)^{-1}$. Given $j\in\{1,2\}$, let $j'\in\{1,2\}\setminus\{j\}$ and compute
\begin{align*}
& (H_0+1)^{-1} (-i\partial_j -\tilde A_j) - (-i\partial_j -\tilde A_j) (H_0+1)^{-1} \\
& \quad = (H_0+1)^{-1} \left[ -i\partial_j -\tilde A_j, (-i\partial_{j'} -\tilde A_{j'})^2 \right] (H_0+1)^{-1} \\
& \quad = -i(-1)^{j} (H_0+1)^{-1} \left(  \tilde B(-i\partial_{j'} -\tilde A_{j'}) + (-i\partial_{j'} -\tilde A_{j'})\tilde B \right) (H_0+1)^{-1} \,.
\end{align*}
Thus
\begin{align*}
(H_0+1)^{-1} \tilde W \tilde\eta = 
& - (-i\partial_1 -\tilde A_1)(H_0+1)^{-1}(-i\partial_2 -\tilde A_2) q'\tilde\eta - (-i\partial_2 -\tilde A_2)(H_0+1)^{-1}(-i\partial_1 -\tilde A_1)q'\tilde\eta \\
& + (-i\partial_2 -\tilde A_2)(H_0+1)^{-1}(-i\partial_2 -\tilde A_2) (q')^2 \tilde\eta -i(H_0+1)^{-1}(-i\partial_2 -\tilde A_2) q''\tilde\eta \\
& - i (H_0+1)^{-1} \left(  \tilde B(-i\partial_2 -\tilde A_2) + (-i\partial_2 -\tilde A_2)\tilde B \right) (H_0+1)^{-1}  (-i\partial_2 -\tilde A_2) q'\tilde\eta \\
& + i (H_0+1)^{-1} \left(  \tilde B(-i\partial_1 -\tilde A_1) + (-i\partial_1 -\tilde A_1)\tilde B \right) (H_0+1)^{-1}  (-i\partial_1 -\tilde A_1)q'\tilde\eta \\
& -i (H_0+1)^{-1} \left(  \tilde B(-i\partial_1 -\tilde A_1) + (-i\partial_1 -\tilde A_1)\tilde B \right) (H_0+1)^{-1} (-i\partial_2 -\tilde A_2) (q')^2 \tilde\eta
\end{align*}
Now we use the fact that $\| (-i\partial_j - \tilde A_j)^\alpha (H_0+1)^{-1} (-i\partial_k - \tilde A_k)^\beta \|\leq 1$ for any $\alpha,\beta\in\{0,1\}$ and $j,k\in\{1,2\}$. Since $|B|\leq M$ everywhere and since $|q'|\leq M(1+M/4) h^{1/2}\ell$ (see \eqref{eq:qlocbound}) and $|q''| \leq M h^{1/2}$ on the support of $\eta$ we conclude that
$$
\| (H_0+1)^{-1} \tilde W \tilde \eta \| \leq C h^{1/2} \ell
$$
with a constant $C$ depending only on $M$.

Finally, we bound the operator $\tilde g (H_0+1)^{-1} \tilde W (1-\tilde\eta)$. Our goal is to apply Proposition \ref{prop4}, which will give us an exponentially small error with respect to $\ell$. In order to apply this proposition, we have to bound the operator norm of
$$
e^{\delta\langle\cdot-\x_0\rangle} (H_0+1)^{-1} \tilde W (1-\tilde\eta) e^{-\delta\langle\cdot-\x_0\rangle}
$$
in terms of $M$. To do so, we note that $\tilde W (1-\tilde\eta)$ is given by the same expression as $\tilde W \tilde\eta$, except that $\tilde{\eta}$ is replaced by $1-\tilde\eta$. Hence $(H_0+1)^{-1} \tilde W (1-\tilde\eta)$ has a similar representation as $(H_0+1)^{-1} \tilde W_0$ but with $\tilde\eta$ replaced by $1-\tilde\eta$. The required bound now follows from the bounds on $B$, $q'$ and $q''$ together with Proposition \ref{comm} which states that $\| e^{\delta\langle\cdot-\x_0\rangle} (-i\partial_j - \tilde A_j)^\alpha (H_0+1)^{-1} (-i\partial_k - \tilde A_k)^\beta e^{-\delta\langle\cdot-\x_0\rangle}\|$ is bounded for any $\alpha,\beta\in\{0,1\}$, $j,k\in\{1,2\}$ and $|\delta|\leq \delta_0$.

This concludes the proof of the boundedness of $\tilde g (H_0+1)^{-1} \tilde W$. The bound on $\tilde W (\tilde H - z)^{-1} \tilde g$ is derived similarly and we only sketch the major differences. First, it is more convenient to work on $\mathcal P$ instead of $\R^2_+$ and to consider the operator $W (H-z)^{-1} g$, where $H=(-i\nabla -\A)^2$. Now
$$
W = - (-i\partial_1 - A_1) q' (-i\partial_2 - A_2) - (-i\partial_2 - A_2)q'(-i\partial_1 - A_1) - (-i\partial_2 - A_2)(q')^2 (-i\partial_2 - A_2)
$$
(that is, only the sign of the last term changes) and we decompose $W=\eta W+ (1-\eta)W$. Next, we commute all derivatives to the right and then we commute one derivative through $(H-z)^{-1}$. The argument then is similar to what we have done before. The $z$-dependence is controlled by the resolvent identity in the bound for $\eta W$ and by Proposition \ref{prop4} in the bound for $(1-\eta) W$. This concludes the proof of Proposition \ref{bdry2}.
\end{proof}


\section{Proof of Theorem \ref{theorem2}}\label{sectiuneadoi}

Instead of working with the operator $L_h$ and the fixed domain
$\Omega$, we will dilate the domain as in \eqref{prima9} obtaining
$\Omega_h=\{\x\in \R^2:\; h^{1/2}\x\in\Omega\}$. Then $L_h$ will be
unitary equivalent with $hH_h$, where $H_h=(-i\nabla-\A_h)^2$ is
defined in $L^2(\Omega_h)$ with Dirichlet boundary conditions, 
the new vector potential becomes 
$\A_h(\x):=h^{-1/2}\A(h^{1/2}\x)$ and the new magnetic field is
$B_h(\x):=B(h^{1/2}\x)$. Thus \eqref{eq:asymp} can be rewritten as:
\begin{equation}\label{aprrilie1}
\Tr f(H_h) = h^{-1} \left( C_0(f) + h^{1/2} C_1(f) + o(h^{1/2}) \right)
\qquad\text{as}\ h\to 0,
\end{equation}
where $C_0(f)$ and $C_1(f)$ are the coefficients given in the introduction. 


\subsection{Partitions and cut-offs}

Our construction will depend on a parameter $\ell$, which will later be chosen as an inverse power of $h$. (Indeed, $\ell=h^{-1/8}$.) However, at this point we prefer to consider $\ell$ as a parameter independent of $h$. We will assume throughout that $\ell\geq 1$, $h\leq 1$ and that $h^{1/2}\ell\leq d$, where $d$ is chosen so large that the domain $\{\x \in \Omega :\ {\rm dist}\{\x,\partial \Omega\} \leq 100 d \}$ can be parametrized in terms of a tangential and a normal coordinates. This is explained below in more detail. The existence of such $d$ follows from the $C^2$ assumption on $\partial\Omega$.

We define for $t>0$:
\begin{equation}\label{margine1}
\Xi_\ell(t):=\left \{\x\in \overline{\Omega_h}\,:\; {\rm
    dist}\{\x,\partial \Omega_h\}\leq t \ell \right \}\,.
\end{equation}
This models a `thin' compact subset of $\Omega_h$, near the boundary, with a volume of order $\mathcal O(h^{-1/2}\ell)$. Because we assumed that $h^{1/2}\ell$ is small enough, all points of $\Xi_\ell(t)$, $t\leq 100$, have unique projections on $\partial\Omega_h$. We also note that if $t_1<t_2$ then $\Xi_\ell(t_1)\subset \Xi_\ell(t_2)$ and:
\begin{align}\label{margine3}
{\rm dist}\{\Xi_\ell(t_1), \overline{\Omega_h\setminus\Xi_\ell(t_2)}\}
\geq (t_2-t_1) \ell \,. 
\end{align}
The subset $\overline{\Omega_h\setminus\Xi_\ell(100)}$ will model the bulk
region of $\Omega_h$, which is still `far-away' from the
boundary. 

\vspace{0.5cm}

\noindent{\bf Covering the bulk region}. ~\\
We choose three functions
$0\leq g_0\leq \tilde{g}_0\leq \tilde{\tilde{g}}_0\leq 1$ 
with the following properties: 
\begin{align}
&{\rm supp}(g_0)\subset \overline{\Omega_h\setminus\Xi_\ell(100)},\quad 
g_0(\x)=1 \;{\rm on}\; \overline{\Omega_h\setminus\Xi_\ell(150)}
,\label{margine5}\\
&{\rm supp}(\tilde{g}_0)\subset \overline{\Omega_h\setminus\Xi_\ell(50)},\quad 
\tilde g_0(\x)=1\;{\rm on}\; \overline{\Omega_h\setminus\Xi_\ell(75)}
,\label{margine6}\\
&{\rm supp}(\tilde{\tilde{g}}_0)\subset \overline{\Omega_h\setminus\Xi_\ell(10)},\quad 
\tilde{\tilde g}_0(\x)=1\;{\rm on}\;  \overline{\Omega_h\setminus\Xi_\ell(20)}
\label{margine6'}.
\end{align}
There is a constant $C$ such that for all $h\in (0,1]$ and all $\ell\in[1,d h^{-1/2}]$ we have:
\begin{align}\label{margine7} 
\max\{\|D^\beta \tilde{\tilde{g}}_0\|_\infty, \|D^\beta
\tilde{g}_0\|_\infty,\|D^\beta g_0\|_\infty\}\leq C \ell^{-|\beta|} 
\quad \text{for all} \ |\beta|\leq 2 \,.
\end{align}
Moreover:
\begin{align}\label{aprrilie2}
g_0\tilde{g}_0=g_0,\quad
\tilde{g}_0\tilde{\tilde{g}}_0=\tilde{g}_0,\quad {\rm dist}\{{\rm
  supp}(D\tilde{g}_0),{\rm supp}(g_0)\}\geq \ell \,.
\end{align}

\vspace{0.5cm}

\noindent{\bf Covering the boundary region}.~\\
We recall our assumption that $\partial\Omega$ is  a finite union of
disjoint regular, simple and closed $C^2$ curves. For simplicity,
let us restrict ourselves to the simply connected situation. As explained before, we assume that $h^{1/2}\ell$ is sufficiently small such that $\Xi_\ell(100)$ can be parametrized using a tangential
coordinate $\sigma$ living on the torus 
$\frac{h^{-1/2}|\partial\Omega|}{2\pi}S^1$ and a normal 
coordinate $\tau$ in $(0,100 \ell)$. We want to divide this strip in
curvilinear rectangles of side size proportional to 
$\ell$. We define:
\begin{align}\label{adoua1}
M_\ell:=\left [h^{-1/2}\ell |\partial\Omega|\right ]\in
\mathbb{N},\quad \sigma_s:=(s-1)\;\ell\,,\; 1\leq s\leq M_\ell-1\,,\; \sigma_{M_\ell}:=\sigma_1\,.
\end{align}
Note that the (curvilinear) distance between
$\sigma_{M_\ell-1}$ and $\sigma_{M_\ell}(=\sigma_1)$ is at most $2\ell$ and at
least $\ell$. All other curvilinear distances between consecutive points
are $\ell$. Note also that due to the regularity condition on
$\partial\Omega$, the Euclidean distances between these points on
$\partial\Omega_h$ are of order $\ell$.   

Define also:
\begin{align}\label{adoua2}
C_s(t_1,t_2):=\{\x=(\sigma,\tau)\in \Xi_\ell(t_2):\;|\sigma-\sigma_s|\leq t_1 \ell \}\,.
\end{align}
We can now find $M_\ell$ smooth functions $\{g_s\}$ with the following properties:
\begin{align}
&{\rm supp}(g_s)\subset C_s(10, 150)\,;\label{amargine5}\\
&0\leq g_s\leq 1,\qquad \sum_{s=1}^{M_\ell-1}g_s(\x)=1\; {\rm on}\; \Xi_\ell(100)
\,;\label{amargine6}\\
&\|D^\beta g_s\|_\infty\leq C \ell^{-|\beta|} 
\quad \text{for all}\ |\beta|\leq 2\,,\quad  1\leq s\leq M_\ell-1\label{amargine7}\,. 
\end{align}
In order to use a more compact notation, we denote
the set of centers of all curvilinear squares of the type
$C_s(t_1,t_2)$ with $E_{\rm bdy}$. Thus \eqref{amargine6} reads as:
\begin{align}\label{aprrilie4}
0\leq g_\gamma\leq 1,\qquad \sum_{\gamma\in E_{\rm bdy}}
g_\gamma(\x)=1\; {\rm on}\; \Xi_\ell(100).
\end{align}
Another important property is that a point $\x\in \Xi_\ell(100)$ can 
only belong to the support of a finite number of $g_\gamma$'s, and
this number is bounded independently of $h$ and $\ell$.

For every $\gamma\in E_{\rm bdy}$ we can construct two other smooth functions
$0\leq \tilde{g}_\gamma\leq \tilde{\tilde{g}}_\gamma\leq 1$ with the properties:
\begin{align}\label{aprrilie5}
&g_\gamma\tilde{g}_\gamma=g_\gamma,\quad \tilde{g}_\gamma
\tilde{\tilde{g}}_\gamma=\tilde{g}_\gamma, \quad 
\max\{\|D^\beta \tilde{\tilde{g}}_\gamma\|_\infty, \|D^\beta
\tilde{g}_\gamma\|_\infty,\|D^\beta g_\gamma\|_\infty\}\leq C \ell^{-|\beta|} \quad \text{for all}\ 
|\beta|\leq 2,\\
&{\rm dist}\{{\rm supp}(D\tilde{g}_\gamma),{\rm supp}(g_\gamma)\}\geq \ell\,,
\quad {\rm diam}\{{\rm supp}(\tilde{\tilde{g}}_\gamma)\}\leq  C \ell\,.\label{aprrilie6}
\end{align}

\vspace{0.5cm}

 \noindent{\bf Patching the two regions together}.~\\
 It is clear from
 the construction of $g_0$ and the boundary covering that they can be made
 compatible with the supplementary condition that they generate a 
partition of unity of the whole dilated domain
 $\Omega_h$. Thus we can impose: 
\begin{equation}\label{adoua3}
g_0(\x)
+\sum_{\gamma\in E_{\rm bdy}}g_\gamma(\x)=1,\quad \forall \x\in \Omega_h.
\end{equation} 


\subsection{Local model operators}\label{sec:modelop}

Let us define 
\begin{align}\label{adoua12}
\A^{\rm bulk}_h(\x):=\tilde{\tilde{g}}_0(\x)\A_h(\x).
\end{align}
We can define an operator $H^{\rm bulk}_h:=(-i\nabla -\A^{\rm
  bulk}_h)^2$ defined in $L^2(\R^2)$ and see that 
due to the support conditions we have the identity:
\begin{equation}\label{adoua13}
\tilde{g}_0  \A^{\rm bulk}_h= \tilde{g}_0 \A_h \,,
\quad H^{\rm bulk}_h\tilde{g}_0 =H_h\tilde{g}_0 \,.
\end{equation}

Moreover, for every $\gamma\in E_{\rm bdy}$ we define:
\begin{align}\label{aprrilie10}
\A^{(\gamma)}_h(\x):=\tilde{\tilde{g}}_\gamma(\x)\A_h(\x).
\end{align}
For all sufficiently small $h^{1/2}\ell$, the support of each $\tilde{\tilde{g}}_\gamma$ contains a piece of 
$\partial\Omega_h$ having a length of order $\mathcal O(\ell)$ and is, after a translation and rotation, included in a perturbed half-plane $\mathcal{P}^{(\gamma)}$ of the form \eqref{marrtie1}. Here $q^{(\gamma)}$ coincides with the local parametrization of $\partial\Omega_h$. We now define the operator $H_h^{(\gamma)}=(-i\nabla -\A^{(\gamma)}_h)^2$ with Dirichlet boundary conditions in $L^2(\mathcal{P}^{(\gamma)})$. Again due to support considerations, 
\begin{align}\label{aprrilie11}
\tilde{g}_\gamma\A^{(\gamma)}_h =\tilde{g}_\gamma\A_h \,,
\quad H_h^{(\gamma)}\tilde{g}_\gamma=H_h\tilde{g}_\gamma \,.
\end{align}


\subsection{An approximation for $f(H_h)$}

The following lemma tells us that, up to a controlled error, we can compute $\Tr f(H_h)$ by computing $\Tr f(H_h^{\rm bulk})g_0$ and $\Tr f(H_h^{(\gamma)}) g_\gamma$ for $\gamma\in E_{\rm bdy}$. We recall that $d$ is the upper bound on $h^{1/2}\ell$, which is needed in order to introduce boundary coordinates.

\begin{lemma}\label{lemma2}
Given $f\in\mathcal{S}(\R)$ and $N\geq 0$, there is a constant $C>0$ such that for all $h\in (0,1]$ and $\ell\in [1,d h^{-1/2}]$ one has
\begin{align}\label{aprrilie13}
\left\| f(H_h) - \tilde{g}_0f(H_h^{\rm bulk})g_0 - 
\sum_{\gamma\in E_{\rm bdy}} \tilde{g}_\gamma f(H_h^{(\gamma)}) g_\gamma
\right\|_{B_1(L^2(\Omega_h))} 
\leq C \, h^{-7/4} \ell^{-N} \,.
\end{align}
\end{lemma}

Presumably, the order $h^{-7/4}$ in \eqref{aprrilie13} is not optimal, but it is sufficient for our purposes.

\begin{proof}
\emph{Step 1. An approximate resolvent}\\
If $z\in\mathbb{C}$ with $\Im z\neq 0$, we define the operator: 
\begin{align}\label{geigi4}
S_h(z)& :=\tilde{g}_0(H_h^{\rm bulk}-z)^{-1}g_0+
\sum_{\gamma\in E_{\rm bdy} }\tilde{g}_\gamma (H_h^{(\gamma)}-z)^{-1}g_\gamma. 
\end{align}
One can prove that the range of $S_h(z)$ is in the domain of
$H_h$ and we have:
\begin{align}\label{geigi5}
&(H_h-z)S_h(z)=1 +V_h(z),\end{align}
with 
\begin{align}\label{geigi5a}
V_h(z)&:=\left \{-2i(\nabla\tilde{g}_0)\cdot
(-i\nabla -\A_h^{\rm bulk})-(\Delta\tilde{g}_0)\right \} (H_h^{\rm
  bulk}-z)^{-1}g_0\nonumber \\
&+\sum_{\gamma\in E_{\rm bdy} }
\{-2i(\nabla\tilde{g}_\gamma)\cdot
(-i\nabla -\A_h^{(\gamma)})-(\Delta\tilde{g}_\gamma)\} (H_h^{(\gamma)}-z)^{-1}
g_\gamma. 
\end{align}
In order to obtain \eqref{geigi5} we observe that $H_h$ can be
replaced by the local Hamiltonians on the support of the corresponding
cut-off functions, see \eqref{adoua13} and \eqref{aprrilie11}. 
The operator $V_h$ contains only terms coming from
commuting the local Hamiltonians with cut-off
functions. The identity
operator on the right hand side of \eqref{geigi5} appears after the
use of \eqref{adoua3}. 

It follows from \eqref{geigi5} that
\begin{align}\label{apatra1}
(H_h-z)^{-1}=S_h(z)-(H_h-z)^{-1}V_h(z).
\end{align}
Standard estimates (see \eqref{keyest}) yield
for any $h\in (0,1]$ the existence of a constant $C(h)$, such that 
$$
\|V_h(z)\|\leq C(h)\frac{|\Re z|^{1/2}}{|\Im z|}\,, \quad \forall z\in\mathcal{D}\,.
$$
The Helffer-Sj\"ostrand formula \eqref{azecea6} shows that 
for any almost analytic extension $f_{a,N}$ of $f$ satisfying 
\eqref{conditiihesj} with $N\geq 3$ one has:
\begin{align}\label{apatra3}
f(H_h)=\tilde{g}_0f(H_h^{\rm bulk})g_0+\sum_{\gamma\in E_{\rm bdy}}
\tilde{g}_\gamma f(H_h^{(\gamma)})g_\gamma
 -\frac{1}{\pi}\int_{\mathcal{D}} \frac{\partial
  f_{a,N}}{\partial \overline{z}}(H_h-z)^{-1}V_h(z)dxdy.
\end{align}
Hence \eqref{aprrilie13} will follow if we can show that for any 
$N_1>0$ there is an $N_2>0$ and a $C>0$ such that for any 
$h\in (0,1]$ and $\ell\in [1,d h^{-1/2}]$ one has 
\begin{align}\label{apatra4}
\left\| \int_{\mathcal{D}} \frac{\partial
  f_{a,N_2}}{\partial \overline{z}}(H_h-z)^{-1}V_h(z) \,dxdy
\right\|_{B_1(L^2(\Omega_h))} \leq C\,h^{-5/2} \ell^{N_1}\,.
\end{align}

\emph{Step 2. Preliminary estimates}\\
We show that $(H_h-z)^{-1}$ is a Hilbert-Schmidt
operator. First, for $z=-1$ we use the diamagnetic inequality in order 
to bound the absolute value of the integral kernel of $(H_h+1)^{-1}$
with the (positive) integral kernel of $(-\Delta+1)^{-1}$.  
Thus we obtain: 
\begin{align}\label{eq:hs-1}
\left\| (H_h+1)^{-1}\right\| _{B_2(L^2(\Omega_h))}\leq C 
\sqrt{|\Omega_h|}=C \sqrt{|\Omega|}\, h^{-1/2} \,.
\end{align}
For general $z$ we use the resolvent equation
$$(H_h-z)^{-1}=(H_h+1)^{-1}+(z+1)(H_h-z)^{-1}(H_h+1)^{-1}$$
in order to obtain the existence of $C>0$ such that
\begin{align}\label{hcatreia1}
\left\| (H_h-z)^{-1}\right\|_{B_2(L^2(\Omega_h))}\leq C\; h^{-1/2}\; \frac{\langle
 \Re z\rangle }{|\Im z|},
\end{align}
for all  $h$ and $z\in\mathcal D$.

Next, we consider the operators $V_h(z)$ and prove that there are constants $C$ and $\delta$ such that for all $z\in \mathcal{D}$, all $h\in(0,1]$ and all $\ell\in [1,dh^{-1/2}]$ one has
\begin{align}\label{eq:vhinfty}
\left\| V_h(z) \right\| \leq C h^{-3/4} \ell^{-1/2} \frac{\langle \Re z\rangle}
{|\Im z|}\exp\{-\delta \ell |\Im z|/\langle \Re z\rangle\}
\end{align}
and
\begin{align}\label{eq:vh2}
\left\| V_h(z) -V_h(-1) \right\|_{B_2(L^2(\Omega_h))}
\leq 
C h^{-3/4} \ell^{-1/2} \frac{\langle \Re z\rangle^2}
{|\Im z|}\exp\{-\delta \ell |\Im z|/\langle \Re z\rangle\}.
\end{align}
In the following $z$ will always be restricted to $\mathcal{D}$. 
In order to prove \eqref{eq:vhinfty}, let us take a generic term of 
$V_h(z)$, that is 
$(\nabla\tilde{g}_\gamma)\cdot (-i\nabla
-\A_h^{(\gamma)})(H_h^{(\gamma)}-z)^{-1}g_\gamma$. 
We can cover the supports of 
$\nabla\tilde{g}_\gamma$ and $g_\gamma$
with approximately $\ell^2$ disjoint squares of diameter $1$. 
Denote one such square covering the support of
$\nabla\tilde{g}_\gamma$ by $\Lambda_1$, and one such
square covering the support of $g_\gamma$ by $\Lambda_2$. 
Let $\chi_j$, $j=1,2$ denote their characteristic functions. Due to
the condition \eqref{aprrilie6}, the distance between
$\Lambda_1$ and $\Lambda_2$ is at least of order $\ell$. 

Now $H_h^{(\gamma)}$ plays the role of $H_D$ in \eqref{hcadoua7}, and  
$(\nabla\tilde{g}_\gamma)\cdot (-i\nabla -\A_h^{(\gamma)})$ is $W$ in the same
estimate. A consequence of Proposition \ref{prop4} and of
\eqref{aprrilie5} is the estimate:
$$
\left\| \chi_1 (\nabla\tilde{g}_\gamma)\cdot (-i\nabla
 -\A_h^{(\gamma)})(H_h^{(\gamma)}-z)^{-1}\chi_2\right\|
 \leq C \ell^{-1}
 \frac{\langle \Re z\rangle}{|\Im z|}\exp\{-\delta \ell |\Im z|/\langle \Re z\rangle\} \,, 
$$
where the constants $C$ and $\delta$ are independent of $h$, $\gamma$,
$z$ 
and the positions of $\chi_1$ and $\chi_2$.

A similar estimate holds for 
$\chi_1 (\Delta\tilde{g}_\gamma) (H_h^{(\gamma)}-z)^{-1}\chi_2$. 
Now let $\chi_1^{(k)}$, $k=1,\ldots,K$, be a family of characteristic 
functions corresponding to disjoint squares of diameter one covering 
$\supp(\nabla \tilde g_\gamma)$, and similarly $\chi_2^{(l)}$,
$l=1,\ldots,L$, covering $\supp g_\gamma$. Then for all $\Psi,\Phi\in L^2(\Omega_h)$
\begin{align*}
& \left| \left\langle \Psi, \left\{-2i(\nabla\tilde{g}_\gamma)\cdot
(-i\nabla -\A_h^{(\gamma)})-(\Delta\tilde{g}_\gamma)\right\}
(H_h^{(\gamma)}-z)^{-1}g_\gamma \ \Phi\right \rangle \right| \\
& \quad \leq \sum_{k,l} \left| \left \langle \Psi, 
\chi_1^{(k)} \left\{-2i(\nabla\tilde{g}_\gamma)\cdot
(-i\nabla -\A_h^{(\gamma)})-(\Delta\tilde{g}_\gamma)\} 
(H_h^{(\gamma)}-z)^{-1}g_\gamma \right\}
\chi_2^{(l)} \Phi\right\rangle \right| \\
& \quad \leq C \ell^{-1} \frac{\langle \Re z\rangle}{|\Im z|}\exp\{-\delta \ell |\Im z|/\langle \Re z\rangle\}
\sum_{k,l} \left\|\chi_1^{(k)} \Psi \right\| \left\|\chi_2^{(l)} \Phi \right\|.
\end{align*}
Since $K$ and $L$ are bounded by a constant times $\ell^2$, we have 
$\sum_{k} \left\|\chi_1^{(k)} \Psi \right\|\leq C \ell
\|\Psi\|$ by the Schwarz inequality and similarly for $\Phi$. From this we conclude that
$$
\left\|\left\{-2i(\nabla\tilde{g}_\gamma)\cdot
(-i\nabla -\A_h^{(\gamma)})-(\Delta\tilde{g}_\gamma)\right\}
(H_h^{(\gamma)}-z)^{-1}g_\gamma \right\|
\leq C \frac{\ell\ \langle \Re z\rangle}{|\Im z|}
e^{-\delta \ell |\Im z|/\langle \Re z\rangle } \,.
$$
Summing over $\gamma$ and using the fact that number of $\gamma$'s
which are non-zero at a given point is uniformly bounded in $h$ and $\ell$ we arrive at
\eqref{eq:vhinfty}. The estimate for the term containing $H^{\rm bulk}_h$ is completely analogous, with the difference that the
number of unit squares needed to cover the support of $g_0$ is of order $h^{-1}$, whereas the number of unit squares needed to cover the support of $\nabla\tilde g_0$ is of order $h^{-1/2}\ell$. This explains the factor $h^{-3/4}\ell^{-1/2}$ in \eqref{eq:vhinfty}. The factor $\ell$ does not appear in \eqref{eq:vhinfty} because of our assumption that $h^{1/2}\ell \leq d$. (It is the bound on $\left\{-2i(\nabla\tilde{g}_0)\cdot (-i\nabla -\A_h^{\rm bulk})-(\Delta\tilde{g}_0)\right\}
(H_h^{\rm bulk}-z)^{-1}g_0$, which is not optimal. One should take into account the decay between the supports of $\chi_1^{(k)}$ and $\chi_2^{(l)}$, similarly as we did in the proof of Proposition \ref{bdry}.)

In order to prove \eqref{eq:vh2} we proceed similarly as before by 
controlling the Hilbert-Schmidt norm of the operator:
\begin{align*}
& \chi_1 (\nabla\tilde{g}_\gamma)\cdot (-i\nabla
-\A_h^{(\gamma)})\{(H_h^{(\gamma)}-z)^{-1}-(H_h^{(\gamma)}+1)^{-1}\}\chi_2 \\
& \quad=
(z+1) \chi_1 (\nabla\tilde{g}_\gamma)\cdot (-i\nabla -\A_h^{(\gamma)})
(H_h^{(\gamma)}+1)^{-1}(H_h^{(\gamma)}-z)^{-1}\chi_2
\end{align*}
with $\chi_1$ and $\chi_2$ as above. Now using the result \eqref{eee2} of Proposition \ref{prop4} we arrive at:
$$
\left\| \chi_1 (\nabla\tilde{g}_\gamma)\cdot (-i\nabla
  -\A_h^{(\gamma)})(
H_h^{(\gamma)}+1)^{-1}(H_h^{(\gamma)}-z)^{-1}\chi_2 \right\|_{B_2}
\leq C \ell^{-1} \frac{\langle \Re z\rangle}{|\Im z|}
e^{-\delta \ell |\Im z|/\langle \Re z\rangle} \,, 
$$
where the constants $C$ and $\delta$ are independent of $h$, $\gamma$,
$z$ and the positions of $\chi_1$ and $\chi_2$. 
Hence, with the same notation as above,
\begin{align*}
& \left\| (\nabla\tilde{g}_\gamma)\cdot (-i\nabla -\A_h^{(\gamma)})
(H_h^{(\gamma)}+1)^{-1}(H_h^{(\gamma)}-z)^{-1} g_\gamma \right\|_{B_2}^2 \\
& \quad =\sum_{k,l}
\left\| \chi_1^{(k)} (\nabla\tilde{g}_\gamma)\cdot (-i\nabla
  -\A_h^{(\gamma)})(H_h^{(\gamma)}+1)^{-1}(H_h^{(\gamma)}-z)^{-1}
  g_\gamma \chi_2^{(l)} \right\|_{B_2}^2 \\
& \quad \leq C \ell^2 \frac{\langle \Re z\rangle^2}{|\Im z|^2}\exp\{-2\delta \ell |\Im z|/\langle \Re z\rangle\}
\end{align*}
Here we used that both the $k$ and the $l$ sum contain 
$\mathcal O(\ell^2)$ terms. Summing over $\gamma$ 
and using again that the overlap of the covering is uniformly bounded, 
we arrive at an estimate as in \eqref{eq:vh2}. The term with $H_h^{\rm
  bulk}$ can be treated in a similar way, where we note that we need
about $h^{-1}$ unit squares to cover the support of $g_0$ and $h^{-1/2}\ell$ to cover that of $\nabla\tilde g_0$. (Again, the argument can be improved at this point.)

\emph{Step 3. Proof of \eqref{apatra4}.}\\
Due to the presence of the momentum operator in the expression of $V_h(z)$, we cannot hope to prove that $V_h(z)$ is Hilbert-Schmidt. We resolve this problem in a similar way as we did in the proof of Proposition \ref{bdry}, namely by replacing the operator $(H_h-z)^{-1}V_h(z)$ under the integral in \eqref{apatra4} by the difference $$(H_h-z)^{-1}V_h(z)-(H_h+1)^{-1}V_h(-1)$$
without changing the integral. In this way, we obtain a sum of two terms:
\begin{equation}\label{hcatreia2}
\left[(H_h-z)^{-1}-(H_h+1)^{-1}\right]V_h(z)+(H_h+1)^{-1}\left[V_h(z)-V_h(-1)\right].
\end{equation}
In the first term we can apply the resolvent formula and arrive at the
operator:
\begin{equation}\label{hcatreia3}
(z+1)(H_h-z)^{-1}(H_h+1)^{-1}V_h(z).
\end{equation}
By \eqref{eq:hs-1} and \eqref{hcatreia1} the product
$(H_h-z)^{-1}(H_h+1)^{-1}$ is a trace class operator with trace norm
bounded by $ C\; h^{-1}\langle \Re z \rangle /|\Im z|$. The bound
\eqref{eq:vhinfty} on the operator norm of $V_h(z)$ 
gives
$$
\left\| (z+1)(H_h-z)^{-1}(H_h+1)^{-1}V_h(z)\right\|_{B_1}
\leq C h^{-7/4}\ell^{-1/2} \frac{\langle
\Re z\rangle^3}{|\Im z|^2}\exp\{-\delta \ell |\Im z|/\langle \Re z\rangle\} \,,
$$
which, in view of the inequality $x^N e^{-x} \leq (N/e)^N$, 
can be further bounded as follows:
\begin{equation}\label{hcatreia4}
\left\| (z+1)(H_h-z)^{-1}(H_h+1)^{-1}V_h(z)\right\|_{B_1}
\leq C_N h^{-7/4} \ell^{-N-1/2}  \frac{\langle \Re z\rangle^{3+N}}{|\Im z|^{2+N}} \,.
\end{equation}
Thus we obtained a power-like decay in $h$ at the price of a more
singular behavior in $\Im z$. But choosing $N$ large enough in
\eqref{apatra2} allows us to kill the
singularity in $\Im z$ and -- at the same time -- to  turn the growth 
in $\Re z$ into something integrable. 
Therefore we are done with the first term of
\eqref{hcatreia2}. 

The second term in \eqref{hcatreia2} is bounded using \eqref{eq:hs-1} 
and \eqref{eq:vh2}:
$$
\left\| (H_h+1)^{-1}\left[V_h(z)-V_h(-1)\right] \right\|_{B_1}
\leq C h^{-5/4} \ell^{-1/2} \frac{\langle\Re z\rangle^2}{|\Im z|}
\exp\{-\delta \ell |\Im z|/\langle \Re z\rangle\} \,.$$
As before, this implies
$$ \left\| (H_h+1)^{-1}\left[V_h(z)-V_h(-1)\right] \right\|_{B_1}
\leq C_N h^{-5/4} \ell^{- N-1/2}\frac{\langle\Re z\rangle^{2+N}}{|\Im z|^{1+N}}$$
which completes the proof.
\end{proof}


\subsection{Completing the proof of Theorem \ref{theorem2}}

The bound from Lemma \ref{lemma2} reduces the study of ${\rm Tr}\{f(H_h)\}$ to that of
$$
T_{\rm bulk}(h):= {\rm Tr}\left \{ f(H_h^{\rm bulk}) g_0\right \}
$$
and
$$
T_{\rm bdy}(h):= \sum_{\gamma\in E_{\rm bdy}}{\rm Tr}\left \{ f(H_h^{(\gamma)})g_\gamma \right\} \,.
$$
Indeed, Lemma \ref{lemma2} leads to the estimate: 
\begin{align}\label{apatra5}
{\rm Tr}\{f(H_h)\} = T_{\rm bulk}(h)+T_{\rm bdy}(h)+\mathcal{O}_N(h^{-7/4} \ell^{-N} )
\end{align}
for any $N\geq 1$. Let us briefly mention that by writing
$$
f(H_h^{(\gamma)}) g_\gamma=\{f(H_h^{(\gamma)}) (H_h^{(\gamma)}-i)^2\} (H_h^{(\gamma)}-i)^{-2}g_\gamma
$$
we can apply Proposition \ref{prop15} and conclude that $f(H_h^{(\gamma)}) g_\gamma$ is trace class. The same is true for the
bulk term.

\vspace{0.5cm}

\noindent{\bf The bulk contribution.}~\\
Since the operator $f(H_h^{\rm bulk}) g_0$ is trace class and has a
jointly continuous integral kernel (see Appendix \ref{apendixD}), we
can compute its trace as $T_1(h)=\int_{\R^2}f(H_h^{\rm bulk})(\x,\x)g_0(\x)d\x$. 
The bulk operator $H_h^{\rm bulk}$ is constructed in such a way that
on the support of $\tilde{g}_0$ the vector potential $\A_h^{\rm bulk}$ generates the
magnetic field $B_h(\x)$, which obeys \eqref{eq:bulkass}. Now putting together \eqref{eq:bulk} and 
\eqref{eq:landaudiag} we obtain:
\begin{align}\label{aprrilie20}
T_{\rm bulk}(h)&= (2\pi)^{-1} \sum_{k=1}^\infty \int_{\Omega_h} b_k(B_h(\y),f) B_h(\y) g_0(\y)\,d\y 
+ \mathcal{O}_N( |{\rm supp} g_0 | (h+\ell^{-N}) ) \nonumber \\
& = (2\pi h)^{-1} \sum_{k=1}^\infty \int_{\Omega} b_k(B(\x),f) B(\x) g_0(h^{-1/2}\x)\,d\x
+ \mathcal{O}_N(1+ h^{-1}\ell^{-N}) \,.
\end{align}

\vspace{0.5cm}

\noindent{\bf The boundary contribution.}~\\
The operators $H_h^{(\gamma)}$ are constructed in such a way that the assumptions of Corollary \ref{bdry2cor} are satisfied, with a constant $M$ independent of $\gamma$. The role of $g$ will be played by $g_\gamma$. On the support of $\tilde{\tilde{ g}}_\gamma$ we have introduce a local parametrization $q^{(\gamma)}$ of $\Omega_h$. Since $\Omega_h$ is a dilation of $\Omega$ by a factor $h^{-1/2}$, the functions $h^{1/2} |q^{(\gamma)} {}'|$ and $h |q^{(\gamma)}{}''|$ are bounded on the support of $\tilde{\tilde{ g}}_\gamma$, uniformly in $h$, $\ell$ and $\gamma$. This is, of course, stronger than \eqref{eq:qloccor}. Thus there is a $C$ such that if $h\in (0,1]$, $\ell\in [1,C^{-1} h^{-1/2}]$ and $\gamma\in E_{\rm bdy}$ one has
\begin{align*}
& \left| \Tr \left \{ f(H_h^{(\gamma)})g_\gamma \right\} - \frac1{2\pi h} \sum_{k=1}^\infty \int_{\Omega} g_\gamma(h^{-1/2}\x) b_k(B(\x),f) B(\x) \,d\x \right. \notag\\
& \qquad\qquad\qquad\qquad\left. - \frac1{2\pi h^{1/2}} \sum_{k=1}^\infty \int_{\partial\Omega} g_\gamma(h^{-1/2}\x) s_k(B(\x),f) \sqrt{B(\x)} \,d\sigma(\x) \right| \\
& = \left| \Tr \left \{ f(H_h^{(\gamma)})g_\gamma \right\} - \frac1{2\pi} \sum_{k=1}^\infty \int_{\Omega_h} g_\gamma(\y) b_k(B_h(\y),f) B_h(\y) \,d\y \right. \notag\\
& \qquad\qquad\qquad\qquad\left. - \frac1{2\pi} \sum_{k=1}^\infty \int_{\partial\Omega_h} g_\gamma(\y) s_k(B_h(\y),f) \sqrt{B_h(\y)} \,d\sigma(\y) \right| \\
& \leq C \ell^2 (h^{1/2}\ell^2 + h\ell^4 + \ell^{-2} ) \,.
\end{align*}
Since there are about $h^{-1/2}\ell^{-1}$ elements in $E_{\rm bdy}$ we obtain
\begin{align}\label{asasea8}
 &\left \vert T_{\rm bdy}(h) - \frac1{2\pi h} \sum_{k=1}^\infty \int_{\Omega} g_\gamma(h^{-1/2}\x) b_k(B(\x),f) B(\x) \,d\x \right. \notag\\
 & \qquad\qquad\left. - \frac1{2\pi h^{1/2}} \sum_{k=1}^\infty \int_{\partial\Omega} g_\gamma(h^{-1/2}\x) s_k(B(\x),f) \sqrt{B(\x)} \,d\sigma(\x) \right| \leq C h^{-1/2} \ell (h^{1/2}\ell^2 + h\ell^4 + \ell^{-2} ) \,.
\end{align}
In order to optimize the error, we now choose $\ell=h^{-1/8}$. (Note that with this choice $h^{1/2}\ell$ is arbitrarily small if $h$ is small.) With this choice the right side of \eqref{asasea8} is bounded by a constant times $h^{-3/8}$.

Finally, by choosing $N$ in \eqref{apatra5} and \eqref{aprrilie20} sufficiently large and recalling \eqref{adoua3}, we arrive at
$$
\Tr f(H_h) = h^{-1} \left( C_0(f) + h^{1/2} C_1(f) + O(h^{5/8}) \right)
\qquad\text{as}\ h\to 0 \,.
$$
This completes the proof of Theorem \ref{theorem2}.


\section{Proof of Corollary \ref{theorem1}}\label{sectiuneasase}

This corollary is a rather direct consequence of Theorem
\ref{theorem2}. We will work with the dilated domain, see
\eqref{prima9}. Clearly, because $L_h$ is unitarily equivalent with
$hH_h$ we obtain that
$N(hE,L_h)$ will equal the number of 
eigenvalues of $H_h$ (including multiplicities) which are less or
equal than $E$. In order to simplify things we assume that $K=1$, 
which means that $B_{\rm max}<E<3B_{\rm min}$. 

 Theorem \ref{theorem2} can now be applied for a function  $f\in C_0^\infty(\R)$ which equals $1$ on the interval $[B_{\rm min}-\epsilon,B_{\rm max}+\epsilon]$ (here $\epsilon>0$ is chosen small enough) and whose support is included in $(-\infty, 3B_{\rm min}-\epsilon]$. In this case only the first Landau level will
contribute to the bulk term, and we can thus state the following 
straightforward consequence of our main theorem:
\begin{align}\label{prima7}
&\lim_{h\searrow 0} h^{1/2}\left \vert {\rm Tr}\{f(H_h)\}-\frac{1}{2\pi h}\int_{\Omega}B(\x)\,d\x\right \vert \nonumber \\
&=\frac{1}{2\pi} \int_{\partial\Omega} d\sigma(\x)
\int_{0}^{\infty}dt \, \sqrt{B(\x)}  \left \{\int_{\R}f(B(\x)) e_1(\xi)) \left |\psi_1(t,\xi)\right |^2d\xi- 1\right \} \,.
\end{align}
Now denote by $f_-$ any function as above whose support is included in $(-\infty, E]$, and by $f_+$ any function as above which equals $1$ on $(-\infty, E]$. By a straightforward monotonicity argument we have: 
\begin{align}\label{finalaa1}
& \frac{1}{2\pi} \int_{\partial\Omega} d\sigma(\x)
\int_{0}^{\infty}dt \, \sqrt{B(\x)}  \left \{\int_{\R}f_-(B(\x) e_1(\xi)) \left |\psi_1(t,\xi)\right |^2d\xi- 1\right \} \nonumber \\
& \leq \liminf_{h\searrow 0} h^{1/2}\left \vert N(E,H_h)-\frac{1}{2\pi h}\int_{\Omega}B(\x)\,d\x\right \vert
\ \leq \ \limsup_{h\searrow 0} h^{1/2}\left \vert N_h(hE)-\frac{1}{2\pi h}\int_{\Omega}B(\x)\,d\x\right \vert \nonumber \\
&\leq \frac{1}{2\pi} \int_{\partial\Omega} d\sigma(\x)
\int_{0}^{\infty}dt \, \sqrt{B(\x)}  \left \{\int_{\R}f_+(B(\x) e_1(\xi)) \left |\psi_1(t,\xi)\right |^2d\xi- 1\right \} \,.
\end{align}
At this moment we can take the supremum over $f_-$ and the infimum over $f_+$ and show that they are equal. The main idea is to show that the Lebesgue measure of the set of $\xi$'s near $B(\x) e_1(\xi)=E$ is small, uniformly in $\x\in\partial\Omega$. In fact, using Hadamard's identity \eqref{hadamard} one can prove that uniformly in $\x\in\partial\Omega$:
$$
\lim_{\delta \searrow 0}\left \vert \{\xi\in \mathbb{R}: E-\delta\leq B(\x) e_1(\xi)\leq E+\delta\}\right \vert =0 \,.
$$
The last ingredient is to use that $\psi_1$ has norm one. The proof is over.
\qed


\appendix

\section{Analysis of the one-dimensional model operator}\label{AppendixA}

\subsection{Eigenvalues and eigenfunctions of the model operator}

Recall that the eigenfunctions $\phi_k$ and $\psi_k(\cot,\xi)$ were introduced in the introduction. Here we collect some information on these functions as well as on the corresponding eigenvalues.

\begin{lemma}\label{1.4bis}
There are constants $\alpha>0$ and $C>0$ such that, for all $k\geq 1$,  
\begin{align}\label{azecea41}
&|\psi_k(t,\xi)|\leq C \sqrt{e_k(\xi)+1}\,,\quad |\phi_k(t)|\leq C\;
\sqrt k,\quad t>0,\xi\in\mathbb R,\\
&e_k(\xi)\geq (2k-1),\quad \forall \xi\in \mathbb R\,,\label{azecea43}\\
&e_k(\xi)\geq (4k-1)+\xi^2,\quad \xi\geq 0\,,\label{azecea46}\\
&|\phi_k(t)|\leq C\; e^{-\alpha t^2},\quad |t| \geq C k^\frac 12\,, 
\label{azecea45}\\
&|e_k(\xi) -(2k-1)|\leq C e^{-\alpha |\xi|^2},\quad  \xi \leq -C k^\frac 12,\label{azecea44}\\
&|\psi_k(t,\xi)-\phi_k(t+\xi)|\leq C\; e^{-\alpha (t^2 +
  \xi^2)}\,,\; t>0\;{\rm and}\; \xi \leq -C k^\frac12\,. \label{azecea42}
\end{align}
\end{lemma}

\begin{remark}
The main observation is that $\psi_k$ lives in the classical region
 $(t+\xi)^2 \leq e_k(\xi)$. Hence outside a sufficiently
large neighborhood of this classical region, the eigenfunction should decay
 exponentially. The second point is that if $0$ (the boundary point) is far from the
classical region, then $\phi_k(t+\xi;1)$
 furnishes an excellent approximation of $\psi_k$ in large zones
 containing the classical region but avoiding a neighborhood of $0$.
This is what we now have to control by using explicit Agmon
estimates. \end{remark}

\begin{proof}[Proof of Lemma~\ref{1.4bis}]
Let us first prove \eqref{azecea41}. 
We recall that $\psi_k$ and
$\phi_k$ have $L^2$ norm equal to $1$. One immediately obtains after an
integration by parts that
$$
\|\psi_k\|^2_{H^1(\mathbb R_+)} = 1 + \|\psi'_k\|^2 \leq 1+
e_k(\xi)\,.
$$
Using the Sobolev embedding, we get the existence of $C>0$
 such that, for any $k\geq 1$ and $\xi \in \mathbb R$~:
\begin{equation}
\sup_{t>0}|\psi_k(t,\xi)| \leq C \sqrt{1+ e_k(\xi)}\;.
\end{equation}
Similarly, we get
\begin{equation}
\sup_{s\in \mathbb R} |\phi_k(s)|\leq C \sqrt{k}\;.
\end{equation}

The two inequalities \eqref{azecea43} and \eqref{azecea46} are easy by monotonicity
 of the Dirichlet problem with respect to the domain. Let us also note for
 the record the Hadamard formula applied to our Dirichlet problem:
\beq\label{hadamard}
\frac{d}{d\xi}e_k(\xi) = \psi'_k(0,\xi)^2\,.
\eeq

Let us show that \eqref{azecea45} is a consequence of an Agmon estimate
using a weight $e^{\Phi(s)}$ with $$
\Phi(s) = \beta \frac{s^2}{2}\,,\,0<\beta <1\,.
$$
Since 
$
\{D_s^2+s^2 - (2k-1)\} \phi_k (s)=0\, ,
$
we obtain the identity: 
$$
\|(e^{\Phi} \phi_k)'\|^2 =
\int_{\R} \{-s^2+ (\Phi')^2(s) + (2k-1)\} e^{2 \Phi(s)} \phi_k^2(s) ds\;.
$$
or 
\begin{align}\label{nove1}
\|e^{\Phi} \phi_k\|^2_{H^1(\R)}=\int_{\R} \{2k-(1-\beta^2) s^2\}
e^{2 \Phi(s)} \phi_k^2(s) ds\;\leq 2k e^{\frac{2k \beta}{1-\beta^2}},
\end{align}
where the inequality was obtained by taking a rough upper bound of the factor
multiplying $\phi_k^2(s)$ on the region where this factor is nonnegative.
Hence the Sobolev embedding provides a constant $C>0$ such that 
\begin{equation}
|\phi_k (s)| \leq C k e^{\frac{2k
    \beta}{1-\beta^2}}\; e^{-\beta s^2/2},\quad s\in \R.
\end{equation}
Now, for $0<\beta<1$, there exists $C_\beta$ 
such that 
$$C k e^{\frac{2k
    \beta}{1-\beta^2}}e^{-\beta s^2/4}\leq 1,\quad |s|\geq C_\beta
\sqrt{k}\,,$$
and \eqref{azecea45} follows by choosing for example
$\alpha=\beta/4$. 

Also note that \eqref{azecea45} implies the existence of some smaller
$\alpha$ and possibly larger $C$ such that 
\begin{equation}\label{majsup0}
| \phi_k (t+\xi)| \leq e^{-
\alpha\left( 
(t+\xi)^2 + \xi^2\right)} 
\,,\, \forall t \mbox{ s.t. } |t+\xi|\geq \frac{|\xi|}{2}\geq C \sqrt{k}\,.
\end{equation}

Let us show \eqref{azecea44}.  Due to the lower bound
\eqref{azecea43}, we only need to show the upper bound
\beq
e_k(\xi) \leq (2k-1) + C e^{- \alpha |\xi|^2}\,,\, \mbox{ for } \xi \leq - C
\sqrt{k}\,.
\eeq  
For a given $\alpha$, we can always find $C$  such that 
$C e^{- \alpha C^2} <\frac 14$. In particular, we have
$$
(2j-1) + C e^{- \alpha |\xi|^2} \leq (2j +1) - \frac 12\,, \quad \xi
<-C,\quad j\geq 1.
$$
We will show that, for every $1\leq j\leq k$ and for every $\xi<-C\sqrt{k}$, there is in an interval of half length $C e^{-\alpha \xi^2}$ centered at $(2j-1)$ at least one eigenvalue $e_\ell(\xi)$ of the Dirichlet realization of $D_t^2 + (t+\xi)^2$ in $\R_+$. Since \eqref{azecea43} implies that we cannot have more than $k$ eigenvalues in the energy interval $]1, (2k+1)[$, each such constructed $e_\ell(\xi)$ must be alone in its own interval.

We take $\chi(t) \phi_j (t+\xi)$ as the $j$-th quasimode, where $0\leq \chi\leq 1$ is a $C^\infty$ cut-off function equal to $1$ on $[1 ,+\infty)$ and vanishing in $(-\infty,0]$. The control \eqref{azecea45} (which also implies the $H^1$ control) yields
\beq \label{eq:quasimode}
\left\| \left(D_t^2 + (t+\xi)^2 \right) \chi \phi_j (\cdot+\xi) \right\|_{L^2(\R_+)} \leq C e^{-\alpha \xi^2}
\quad\text{and}\quad
\| \chi \phi_j(\cdot +\xi) \|_{L^2(\R_+)} \geq 1- C e^{-\alpha \xi^2}
\eeq
for $\xi< - C \sqrt k$. The spectral theorem now leads to the conclusion. Hence, we have shown that there exist constants $C$ and $\alpha >0$ (after possibly increasing the old $C$ and lowering the old $\alpha$), such that
\beq
|e_j -(2j-1)| \leq  C e^{ - \alpha \xi^2}\,,\, \forall \xi
\leq - C \sqrt{k},\quad 1\leq j\leq k.
\eeq
The bounds \eqref{eq:quasimode} and the spectral theorem also imply that
\beq \label{controlel2ext2}
\|\psi_j - \chi \varphi_j (\cdot + \xi)\|_{L^2(\R_+)}
 \leq C e^{- \alpha \xi^2}\,,\, \forall \xi \leq - C \sqrt k\,,\quad 1\leq j\leq k\,.
\eeq

We now turn to the proof of \eqref{azecea42}. Once we know that $e_k \sim 2k-1$ under the condition $\xi \leq - C \sqrt k$, what we have done for the Gaussian decay of $\varphi_k$ can be also done for any real eigenfunction $\psi_k$ corresponding to $e_k$. Thus for some $\beta >0$ we have
\beq \label{majsup1}
| \psi_k (t )| \leq \widehat C_\beta e^{-\beta |t+\xi|^2}
\,,\, \forall t >0\, \mbox{ s.t. } |t+\xi|\geq \widehat C_\beta
\sqrt{k},\quad \xi\leq -C\sqrt k.
\eeq
We note that under the condition $\xi \leq - C \sqrt{k}$, this implies
also
\beq 
\label{majsup1bis}
| \psi_k (t )| \leq \widehat C_\beta  e^{ - \frac \beta 8 
 (|t+\xi|^2 +\xi^2)} 
\,,\, \forall t >0\, \mbox{ s.t. } |t+\xi|\geq -\frac{\xi}{ 2}\geq C \sqrt{k}/2\,.
\eeq

The bound \eqref{azecea45} together with \eqref{controlel2ext2} implies that
\beq \label{controlel2ext}
\|\psi_k - \varphi_k (\cdot + \xi)\|_{L^2(\R_+)}
 \leq C e^{- \alpha \xi^2}\,,\, \mbox{ if } \xi \leq - C \sqrt k\,.
\eeq
To have a pointwise estimate, we integrate by parts and observe the following inequality:
\begin{align}\label{nove2}
&\|\psi_k' - \varphi_k' (\cdot + \xi)\|^2 + \| (t+\xi) ( \psi_k -
 \varphi_k(\cdot +\xi)\|^2 \nonumber \\ 
& \leq  (2k-1) \|\psi_k - \varphi_k(\cdot +\xi)\|^2 
+ 2 (e_k - (2k-1)) +|\psi_k' (0)- \varphi_k'(\xi )|\;
|\varphi_k(\xi )|\;.
\end{align}
The right hand side is exponentially small like  $\exp(-\alpha \xi^2)$ if we control $\varphi_k'(\xi )$ and $\psi_k' (0)$. This will come from  a
 control of both of them in $H^2(\R_+)$, through $(D_s^2+s^2)$ with
 the appropriate boundary conditions. Hence we have
$$
\| \varphi_k \|_{H^2(\R)}+\| \psi_k \|_{H^2(\R_+)} \leq C \{(2k-1)+e_k\}\,.
$$
Thus
\beq \label{controlel2extb}
\sup_{t>0}|\psi_k(t) - \varphi_k (t + \xi)|
 \leq C e^{ - \alpha \xi^2}\,,\, \mbox{ if } \xi \leq - C \sqrt k\,.
\eeq
We note that the last inequality implies:
\beq\label{majsup2}
|\psi_k(t) - \varphi_k (t + \xi)|
 \leq C e^{- \frac \alpha 2  (\xi^2+ (t+\xi)^2)}\,,\, \mbox{ if }
 \xi \leq - C \sqrt k \mbox{ and } |t+\xi|\leq \frac{|\xi|}{2}\,.
\eeq

We are now able to complete the proof of \eqref{azecea42}. This is simply a
consequence
 of \eqref{majsup0}, \eqref{majsup1bis} (when $ 0< t < \frac {-\xi}{2}
 $
and $t> \frac{-3\xi}{2}$)
 and of \eqref{majsup2} when $\frac{-\xi}{2} \leq t \leq
 \frac{-3\xi}2$.
More precisely, we have shown the existence of $C$ and $\alpha >0$
such that
\beq\label{majsupdef}
|\psi_k(t) - \varphi_k (t + \xi)|
 \leq C e^{- \frac \alpha 2  (\xi^2+ (t+\xi)^2)}\,,\, \mbox{ if }
 \xi \leq - C \sqrt k \mbox{ and } t >0\,.
\eeq
This finishes the proof of Lemma~\ref{1.4bis}.
\end{proof}

\subsection{Proof of Lemma \ref{salternative}}\label{AppendixB}

Replacing $f(E)$ by $f(B^{-1}E)$, we may assume in the proof that $B=1$. Since $\phi_k$ is normalized we have
$$
\int_\R |\phi_k(t+\xi)|^2 \,d\xi = 1 \,,
$$
and therefore we can rewrite $s_k(1,f)$ in the more symmetric form
\begin{equation}
\label{eq:azecea4''}
s_k(1,f) = \int_0^\infty  \left( \int_\R \left( f(e_k(\xi)) |\psi_k(t,\xi)|^2 - f(2k-1)|\phi_k(t+\xi)|^2 \right)\,d\xi \right) \, dt \,.
\end{equation}
In this form, things become clearer. We can use the estimates from Lemma \ref{1.4bis} on the behavior of the eigenvalues $e_k(\xi)$ and eigenfunctions $\psi_k(t,\xi)$ to show that $s_k(1,f)$ is well-defined and summable with respect to $k$. Before we do this, we note that once the convergence of the integrals in \eqref{eq:azecea4''} is proved we can apply Fubini's theorem and interchange the order of integration. Since $\psi_k$ is normalized, this yields the equivalent formula \eqref{eq:salternative} for $s_k(1,f)$. Also the continuous differentiability with respect to $B>0$ is a consequence of our bounds.

We now prove that the integral \eqref{eq:azecea4''} converges and is summable with respect to $k$. We split the integral with respect to $\xi$ in four regions: (I) $\xi\geq C\sqrt k$, (II) $0\leq\xi< C\sqrt k$, (III) $-C\sqrt{k}\geq\xi<0$, and
(IV) $\xi< -C\sqrt{k}$, where $C$ is the constant provided by Lemma \ref{1.4bis}.  
\begin{itemize}
\item We start with region (I). We use the fact that $f$ is rapidly decaying and get that, for any $N$, there is a $C_N>0$
  such that
$$
|f(E)|\leq C_N (1+E)^{-2N}\,.
$$ 
 Then we use \eqref{azecea46} and \eqref{azecea43} and write:
$$
|f[e_k(\xi)]| \leq C_N k^{-N}(1+\xi^2)^{-N} \,.
$$
($C_N$ here denotes a generic $N$-dependent constant, which may vary from line to line.) This when multiplied by $|\psi_k(t,\xi)|^2$ is absolutely integrable in $\xi$ and $t$, and the resulting integral is
summable in $k$ if $N>1$.\\
For the second term we use \eqref{azecea45}:
$$
\int_0^\infty dt \int_{C\sqrt{k}}^\infty d\xi |f(2k-1)|\left
  |\phi_k(t+\xi)\right |^2\leq C |f(2k-1)|\int_0^\infty dt
\int_{0}^\infty e^{-\alpha t^2}e^{-\alpha \xi^2}\leq C_N k^{-N} \,,
$$
which again is summable if $N>1$.

\item We continue with region (II). The first term, $f(e_k(\xi)) |\psi_k(t,\xi)|^2$, is treated as in region (I). For the second term we use the fact that $\phi_k$ is normalized to get
$$
\int_0^\infty dt \int_0^{C\sqrt{k}}d\xi |f(2k-1)|\left |\phi_k(t+\xi)\right |^2\leq C |f(2k-1)|\sqrt{k} \leq C_N k^{-N+1/2}
$$
which is summable for $N>3/2$.

\item In region (III) we use \eqref{azecea43} and see that both $f(e_k(\xi))$ and $f(2k-1)$ are bounded from above by a constant times $k^{-N}$. The integral with respect to $\xi$ gives just an extra $\sqrt{k}$, while the integral in $t$ is taken care of by $\psi_k$ and $\phi_k$.

\item Finally, in region (III) we split the integrand in the following way:
\begin{align*}
&f(e_k(\xi))
\left |\psi_k(t,\xi)\right |^2 -f(2k-1)\left
  |\phi_k(t+\xi)\right |^2\nonumber \\
&=\left( f(e_k(\xi))-f(2k-1)\right) \left |\psi_k(t,\xi)\right |^2 +f(2k-1)\left( \psi_k(t,\xi) +\phi_k(t+\xi) \right)\left(\psi_k(t,\xi)  -\phi_k(t+\xi)\right)\,. 
\end{align*}
We treat the two terms on the right side separately.\\
Let us start with $\left(f(e_k(\xi))-f(2k-1)\right)\left|\psi_k(t,\xi)\right |^2$. Since we are in the region where
\eqref{azecea44} applies, we can, by using the mean value theorem and the fast decay of $f'$, bound the double integral by:
$$
C_N k^{-N}\int_0^\infty dt \int_{-\infty} ^{-C\sqrt{k}}d\xi e^{-\alpha \xi^2}\left |\psi_k(t,\xi)\right |^2\leq \tilde C_N k^{-N} \,.
$$
This takes care of the first term on the right side. For the second term we use \eqref{azecea41} and obtain:
$$
\left|f(2k-1) \left(\psi_k(t,\xi) +\phi_k(t+\xi) \right)\right| \leq C_N k^{-N+1/2} \,.
$$
Then \eqref{azecea42} provides the absolute integrability of $\psi_k(t,\xi)-\phi_k(t+\xi)$ on the whole region, whose integral is uniformly bounded in $k$. 
\end{itemize}
This concludes the proof of Lemma \ref{salternative}.

\begin{lemma}\label{moment}
Let $f$ be a Schwartz function on $[0,\infty)$. For $B>0$ the sum
$$
\sum_{k=1}^\infty \int_0^\infty  t \left| \int_\R \left( f(B e_k(\xi)) |\psi_k(t,\xi)|^2 - f(B(2k-1)) \right)\,d\xi \right| \, dt
$$
converges.
\end{lemma}

\begin{proof}
The proof is similar to that of Lemma \ref{salternative}. Again we can multiply $f(B(2k-1))$ by $|\phi_k(t+\xi)|^2$ without changing the integral. We split the $\xi$-integral into the same four regions. The only new ingredient is the bound
\begin{align*}
\int_0^\infty t |\psi_k(t,\xi)|^2 \,dt & = \int_0^\infty (t+\xi) |\psi_k(t,\xi)|^2 \,dt -\xi \int_0^\infty |\psi_k(t,\xi)|^2 \,dt \\
& \leq \left( \int_0^\infty (t+\xi)^2 |\psi_k(t,\xi)|^2 \,dt \right)^{1/2} + \max\{-\xi,0\} \\
& \leq \sqrt{e_k(\xi)} + \max\{-\xi,0\}
\end{align*}
and the corresponding bound
$$
\int_0^\infty t |\phi_k(t+\xi)|^2 \,dt \leq \sqrt{2k-1} + \max\{-\xi,0\} \,.
$$
We omit the details.
\end{proof}



\section{Uniform exponential estimates}

\subsection{Integral bounds on the Green's function}

In this subsection we assume that $\Lambda\subset\R^2$ is an open set and that ${\bf A} =(A_1,A_2)\in L^{1,{\rm loc}}(\Lambda,\R^2)$ is a magnetic vector potential. We denote by $H_D:=(-i\nabla -{\bf A})^2$ the operator corresponding to the closure of the quadratic form $\int_\Lambda |(-i\nabla -{\bf A})\psi|^2 \,d\x$ initially defined on $C_0^\infty(\Lambda)$. It is remarkable that the results in this subsection do not require any smoothness of $\A$ or $\partial\Lambda$. We emphasize the crucial fact that \emph{all constants below can be chosen to be independent of $\Lambda$ and ${\bf A}$}.

Throughout this section we abbreviate
$$
\mathcal D := \{ z\in\C :\ 0<|\Im z|\leq 1 \}
$$
and
$$
r:= \sqrt{\langle \Re z\rangle}\,,
\qquad \eta:=|\Im z| 
$$
for $z\in\C$. Moreover, if $\x_0$ is a fixed point in
$\Lambda$ and $\alpha\in \R$, we denote by $e^{\alpha \langle \cdot
  -\x_0\rangle }$ the multiplication operator with the function
$\x\mapsto e^{\alpha \sqrt{ |\x -\x_0|^2+1}}$. The following proposition is the key to all the other estimates in the section.

\begin{proposition}\label{prop2}
There is a $\delta_0>0$ such that for all $\x_0\in\Lambda$, all $z\in\mathcal D$ and all $\delta\in(0,\delta_0]$ one has
 \begin{align}\label{expdek}
\left \Vert 
e^{{\frac{\pm\delta \eta}{r }}\langle \cdot -\x_0\rangle }
(H_D-z)^{-1}e^{{\frac{\mp\delta \eta}
{r}}\langle \cdot -\x_0\rangle }\right
\Vert \leq \frac{2}{\eta}
\end{align}
and
\begin{align}\label{expdek2}
&
  \left \Vert (H_D+1)
e^{\frac{\pm\delta \eta}{r}{\langle\cdot -\x_0\rangle}}(H_D-z)^{-1}
e^{\frac{\mp\delta \eta}{r }\langle \cdot -\x_0\rangle}\right
\Vert \leq 6 \frac{r^2}{\eta} \,.
\end{align}
\end{proposition}

Its proof relies on the following elementary

\begin{lemma}\label{expdecay1}
For any $z\in\C$ and any $j\in\{1,2\}$ we have 
\begin{align}\label{keyest}
\left\| (-i\nabla_j -A_j) (H_D-z)^{-1}\right\| \leq \sqrt{1/|\Im z|+\max\{\Re z,0\}/|\Im z|^2} \,.
\end{align}
\end{lemma}

Indeed, this is an immediate consequence of the following identity, valid for every $\psi\in L^2(\Lambda)$,
\begin{align*}
\sum_{j=1}^2\left\|(-i\nabla_j -A_j) (H_D-z)^{-1}\psi\right\|^2 
=\Re \langle  (H_D-z)^{-1}\psi,\psi \rangle   +\Re z\,\left\|(H_D-z)^{-1}\psi\right\|^2 \,.
\end{align*}

\begin{proof}[Proof of Proposition \ref{prop2}]
Since multiplication with the exponential weight is an unbounded operator and does not leave
invariant the domain of $H_D$, we must work with regularized weights 
of the type $e^{\alpha h_\epsilon(\x)}$ where:
\begin{equation}\label{hcprima3}
h_\epsilon(\x):=\frac{\sqrt{ |\x -\x_0|^2+1}}{\sqrt{ \epsilon |\x
    -\x_0|^2+1}},\quad \epsilon >0.
\end{equation} 
One should note that $h_\epsilon$ is bounded (though not uniformly in
$\epsilon$), while derivatives of any order of $h_\epsilon$ are not
only bounded, but uniformly bounded 
in $\epsilon$ and $\x_0$. In particular, $e^{\alpha
  h_\epsilon(\x)}$ leaves invariant 
the domain of $H_D$.

For any $s >0$ the Combes-Thomas rotation \cite{CT} reads 
$$
e^{s h_\epsilon }(H_D-z)e^{-s h_\epsilon}
=H_D-z+ 2is \nabla h_\epsilon\cdot [-i\nabla-A]+s\Delta h_\epsilon-s^2 |\nabla h_\epsilon|^2 \,.
$$
Since first and second derivatives of $h_\epsilon$ are uniformly bounded in $\epsilon$ and $\x_0$, the bound \eqref{keyest} yields for $s=\delta \eta /r $ with $0<\delta\leq 1$ the estimate
$$
\left\|\left\{ 2is \nabla h_\epsilon\cdot [-i\nabla-A]+s\Delta h_\epsilon-s^2 |\nabla h_\epsilon|^2\right\}(H_D-z)^{-1}\right\|
\leq C \delta \,.
$$
Here $C$ is a constant independent of $z$, $\delta$, $\epsilon$ and $\x_0$. 
Thus if we choose $\delta \leq 1/(2C)$ we may write
\begin{align}\label{expdek4}
&e^{s h_\epsilon }(H_D-z)^{-1}e^{-s h_\epsilon }=(H_D-z)^{-1}
\left \{1+ \left[ 2is \nabla h_\epsilon\cdot [-i\nabla-A]+s\Delta h_\epsilon-s^2 |\nabla h_\epsilon|^2  \right] (H_D-z)^{-1}\right \}^{-1}
\end{align}
and deduce $\left\| e^{s h_\epsilon }(H_D-z)^{-1}e^{-s h_\epsilon } \right\| \leq 2/\eta$.
For $\psi_1\in L_{\rm comp}^2(\Lambda)$ and $\psi_2\in L^2(\Lambda)$ with 
$\|\psi_1\|=\|\psi_2\|=1$ this implies
\begin{align*}
|\langle e^{{\frac{\delta \eta}{r }\langle \cdot -\x_0\rangle }}\psi_1, 
(H_D-z)^{-1} e^{-{\frac{\delta \eta}{r }\langle \cdot -\x_0\rangle
  }}\psi_2\rangle |=
\lim_{\epsilon\searrow 0}
|\langle \psi_1,e^{s h_\epsilon }(H_D-z)^{-1}e^{-s h_\epsilon }
\psi_2\rangle |\leq \frac{2}{\eta}.
\end{align*}
By a density argument we see that $(H_D-z)^{-1}
e^{-{\frac{\delta \eta}{r }}\langle \cdot -\x_0\rangle }\psi_2$ is in
the domain of $e^{{\frac{\delta \eta}{r }}\langle \cdot -\x_0\rangle
}$, and then \eqref{expdek} follows.

Regarding \eqref{expdek2}, one uses
$(H_D+1)(H_D-z)^{-1}=1+(z+1)(H_D-z)^{-1}$ on the right hand side of
\eqref{expdek4}, then after a short density 
argument one employs \eqref{expdek} and the proof is over.
\end{proof}

We also need the following small variation of Proposition \ref{prop2}.

\begin{proposition}\label{comm}
There are constants $\delta_0>0$ and $C>0$ such that for all $\x_0\in\Lambda$, all $j,k\in\{1,2\}$, all $\alpha,\beta\in\{0,1\}$ and all $\delta\in(0,\delta_0]$ one has
 \begin{align}\label{eq:comm}
\left \Vert e^{\pm\delta\langle \cdot -\x_0\rangle }
(-i\nabla_j-A_j)^\alpha (H_D+1)^{-1} (-i\nabla_k-A_k)^\beta 
e^{\mp\delta \langle \cdot -\x_0\rangle }\right
\Vert \leq C \,.
\end{align}
\end{proposition}

\begin{proof}
We use the same notation as in the previous proof. Since
\begin{align*}
e^{sh_\epsilon}(-i\nabla_j-A_j)^\alpha &= (-i\nabla_j-A_j+is\nabla_j h_\epsilon)^\alpha e^{sh_\epsilon} \\
\quad\text{and}\quad 
(-i\nabla_k-A_k)^\beta e^{-sh_\epsilon}&=e^{-sh_\epsilon}(-i\nabla_k-A_k +is\nabla_k h_\epsilon)^\beta \,,
\end{align*}
where $\nabla h_\epsilon$ is uniformly bounded, we only need to show that for all sufficiently small $s$, the operator $(H_D+1)^{1/2} e^{sh_\epsilon} (H_D+1)^{-1} e^{-sh_\epsilon}(H_D+1)^{1/2}$ is bounded uniformly in $\epsilon\in(0,1]$. As in the previous proof we can shown that $e^{-sh_\epsilon} (H_D+1) e^{sh_\epsilon} = H_D+1 + W$ with $\| W (H_D+1)^{-1/2} \| \leq C(|s|+s^2)$ uniformly in $\epsilon\in(0,1]$. Thus for $|s|\leq\min\{(4C)^{-1}, (4C)^{-1/2}\}$ we have
$$
e^{sh_\epsilon} (H_D+1)^{-1} e^{-sh_\epsilon} = (H_D+1)^{-1/2} (1 + (H_D+1)^{-1/2} W (H_D+1)^{-1/2} )^{-1} (H_D+1)^{-1/2}
$$
and $\| (1 + (H_D+1)^{-1/2} W (H_D+1)^{-1/2} )^{-1} \| \leq 2$.
\end{proof}

\vspace{0.5cm}

In the following we denote by $\Green(\cdot,\cdot,z)$ the integral kernel 
of the operator $(H_D-z)^{-1}$ for $z\in\mathbb{C}\setminus[0,\infty)$.

\begin{corollary}\label{localizHS}
There are constants $\delta>0$ and $C>0$ such that for all $z\in\mathcal D$ and for all $\x\in\Lambda$
\begin{equation}\label{hcadoua4}
\int_\Lambda d\x'\, e^{{\frac{2\delta \eta}{r }\langle \x-\x'\rangle}} |\Green(\x,\x',z)|^2
\leq C\frac{r^4}{\eta^2} \,.
\end{equation}
\end{corollary}

\begin{proof}
Since $\Green(\x,\x',z)=\overline{\Green(\x',\x,\overline z)}$ it is equivalent to prove that there are constants $\delta>0$ and $C>0$ such that for all $z\in\mathcal D$ and for all $\x'\in\Lambda$
\begin{equation}\label{hcadoua4alt}
\int_\Lambda d\x \,e^{{\frac{2\delta \eta}{r }\langle \x-\x'\rangle}} |\Green(\x,\x',z)|^2
\leq C\frac{r^4}{\eta^2} \,.
\end{equation}
Due to the diamagnetic inequality, $|\Green(\x,\x',-1)|$ is pointwise bounded by the kernel of $(-\Delta_D+1)^{-1}$, which in its turn is bounded by the kernel of the resolvent of the Laplace operator defined in the whole of
$\R^2$. From the properties of the latter kernel we deduce that
\begin{equation}\label{azecea14}
 \sup_{\x\in \Lambda}\int_{\Lambda}d\x \left \vert \Green(\x,\x',-1) \right \vert^2 e^{\langle \x-\x'\rangle} <\infty \,.
\end{equation}
Moreover, by the resolvent identity
\begin{equation}\label{eq:resid}
\Green(\x,\x',z) = \Green(\x,\x',-1) + (z+1) \int_{\R^2} d\y \, \Green(\x,\y,z) \Green(\y,\x',-1) \,.
\end{equation}
Hence
\begin{align*}
&\int_\Lambda d\x \, e^{{\frac{2\delta \eta}{r }\langle \x-\x'\rangle}} |\Green(\x,\x',z)|^2
\leq 2 \int_\Lambda d\x \, e^{{\frac{2\delta \eta}{r }\langle \x-\x'\rangle}} |\Green(\x,\x',-1)|^2 \\
& \qquad \qquad\qquad + 2 |z+1|^2 \int_\Lambda d\x \, e^{{\frac{2\delta \eta}{r }\langle \x-\x'\rangle}} 
\left| \int_{\Lambda} d\y \, \Green(\x,\y,z) \Green(\y,\x',-1) \right|^2 \,.
\end{align*}
By \eqref{azecea14} the first integral on the right side is finite if $\delta\leq 1/2$ (because then also $2\delta\eta/r\leq 1$). If, in addition, we choose $\delta$ not larger than the $\delta_0$ in Proposition \ref{prop2}, then by \eqref{expdek} the second integral on the right side is bounded by
$$
\left(\frac 2\eta\right)^2 \int_\Lambda d\x\, e^{{\frac{2\delta \eta}{r }\langle \x-\x'\rangle}} 
\left| \Green(\x,\x',-1) \right|^2 \,,
$$
which again by \eqref{azecea14} is finite. This proves \eqref{hcadoua4alt}.
\end{proof}

As an application of Corollary \ref{localizHS} we derive a useful trace class criterion.

\begin{proposition}\label{prop15}
There is a constant $C>0$ such that for any set $\Lambda_1 \subset \Lambda$ of diameter $\leq 1$ the operator
$\chi_1(H_D-i)^{-2}$, $\chi_1$ being the characteristic function of $\Lambda_1$, is trace class and satisfies
$$
\left\|  \chi_1 (H_D-i)^{-2}\right\| _{B_1(L^2(\Lambda))}\leq C\,.
$$ 
\end{proposition}

\begin{proof}
We fix $\x_0\in \Lambda_1$ and write:
\begin{align}\label{hcadoua116}
\chi(H_D-i)^{-2} =\chi e^{\delta |\x_0-\cdot |}\left\{ e^{-\frac{\delta}{2} |\x_0-\cdot|}  e^{-\frac{\delta}{2} |\x_0-\cdot|}(H_D-i)^{-1}
e^{{\frac{\delta}{2}|\x_0-\cdot|}}\right \} \left \{e^{-{\frac{\delta}{2}|\x_0-\cdot|}}(H_D-i)^{-1}\right\}.
\end{align}
The operator $\chi e^{\delta |\x_0-\cdot |}$ is bounded by $e^{\delta}$ since $\Lambda_1$ has diameter one. The last two operators are both Hilbert-Schmidt because their integral kernels are square integrable by Corollary \ref{localizHS}, and their Hilbert-Schmidt norms can be bounded independently of $\x_0$. 
\end{proof} 

We now turn to the proper topic of this subsection, namely exponential decay estimates. By this we mean bounds in various norms on operators of the form $\chi_1 W (H_D-z)^{-1}\chi_2$, where $\chi_1$ and $\chi_2$ are characteristic functions of disjoint sets.

\begin{proposition}\label{prop5}
There is a constant $C>0$ such that for any disjoint sets $\Lambda_1,\Lambda_2\subset \Lambda$ of diameter $\leq 1$ and of distance $d:={\rm dist}(\Lambda_1,\Lambda_2)\geq 1$ and for any $z\in\mathcal D$ the operator $\chi_1(H_D-z)^{-1}\chi_2$, $\chi_j$ being the characteristic function of $\Lambda_j$, is Hilbert-Schmidt and satisfies
$$
\Vert \chi_1(H_D-z)^{-1}\chi_2\Vert _{B_2(L^2(\Lambda))} 
\leq C\frac{r^2}{\eta} e^{-{\frac{\delta \eta}{r }(d-1)}}.
$$ 
\end{proposition}

\begin{proof}
Fix $\x_0\in \Lambda_1$. We can write:
\begin{align}\label{hcadoua16}
\chi_1(H_D-z)^{-1}\chi_2 =\chi_1 e^{{\frac{\delta \eta}{r
    }|\x_0-\cdot |}} \left\{ e^{-{\frac{\delta \eta}{r }|\x_0-\cdot|}}(H_D-z)^{-1}
e^{{\frac{\delta \eta}{r }|\x_0-\cdot|}}\chi_2\right \} e^{-{\frac{\delta \eta}{r }|\x_0-\cdot|}}\chi_2.
\end{align}
The operator $\chi_1 e^{{\frac{\delta \eta}{r
    }|\x_0-\cdot |}}$ is bounded by $e^{{\frac{\delta \eta}{r
    }}}$, while $e^{-{\frac{\delta \eta}{r }|\x_0-\cdot|}}\chi_2$ is
bounded by $e^{-{\frac{\delta \eta}{r }d}}$. 
The operator in the middle is Hilbert-Schmidt because its integral
kernel is square integrable, as can be easily
inferred from \eqref{hcadoua4}. 
\end{proof} 

Now let us consider the perturbed case. Assume that $W$ is an operator relatively bounded to $H_D$ with a relative bound less than one, which obeys the following condition: 
\begin{align}\label{hcadoua7}
c_\delta= \sup_{\x_0\in \Lambda}\sup_{|\alpha|\leq \delta} \left 
\Vert e^{{\alpha \langle \cdot -\x_0\rangle }}W(H_D-i)^{-1}e^{-\alpha \langle \cdot -\x_0\rangle }\right \Vert 
< \infty .
\end{align}
The example we have in mind is $W=-i\partial_j-a_j$; a quick
commutation shows that this operator has
the property
$$\sup_{\x_0\in \Lambda}\sup_{|\alpha|\leq \delta} \left
\Vert e^{{\alpha \langle \cdot -\x_0\rangle }}We^{-\alpha \langle
  \cdot -\x_0\rangle }(H_D-i)^{-1}\right \Vert <\infty, $$
but this is not quite as in \eqref{hcadoua7}; in order to show that
\eqref{hcadoua7} holds, we write
\begin{align}\label{hcapatra1}
& e^{{\alpha \langle \cdot -\x_0\rangle
  }}W(H_D-i)^{-1}e^{-\alpha \langle \cdot -\x_0\rangle }=\\
&\left\{e^{{\alpha \langle \cdot -\x_0\rangle }}We^{-\alpha \langle
  \cdot -\x_0\rangle }(H_D-i)^{-1}\right \} \left\{(H_D-i)
e^{{\alpha \langle \cdot -\x_0\rangle }}(H_D-i)^{-1}e^{-\alpha \langle
  \cdot -\x_0\rangle }\right\}\nonumber
\end{align}
and now we can invoke \eqref{expdek} and \eqref{expdek2} to see that $c_\delta<\infty$. 

Returning to general operators $W$ we state:

\begin{proposition}\label{prop4}
There are constants $C$ and $\delta>0$ such that for all $\Lambda_1,\Lambda_2\subset \Lambda$ as in Proposition~\ref{prop5} and all $z\in\mathcal D$ one has
\begin{align}\label{eee10}
& \Vert \chi_1 W(H_D+1)^{-1}\chi_2\Vert\leq c_\delta \,C e^{- \delta(d-1)}, \\
\label{eee1}
& \Vert \chi_1 W(H_D-z)^{-1}\chi_2\Vert\leq c_\delta \,C\frac{r^2}{\eta} e^{-{\frac{\delta \eta(d-1)}{r }}}, \\
\label{eee20}
&\Vert \chi_1 W(H_D-i)^{-1}(H_D-z)^{-1}\chi_2\Vert_{B_2}
\leq c_\delta \,C\frac{r^2}{\eta} e^{-{\frac{\delta \eta (d-2)}{2r }}}, \\
\label{eee2}
&\Vert \chi_1 W(H_D+1)^{-1}(H_D-z)^{-1}\chi_2\Vert_{B_2}
\leq c_\delta \,C\frac{r^2}{\eta} e^{-{\frac{\delta \eta (d-2)}{2r }}}.
\end{align}
Here $c_\delta$ is the constant from \eqref{hcadoua7}.
\end{proposition}

\begin{proof}
Denote by $T_1(z):=(H_D-i)(H_D-z)^{-1}=1+(z-i)(H_D-z)^{-1}$ and write
$$
\chi_1W(H_D-z)^{-1}\chi_2=\chi_1 W(H_D-i)^{-1}T_1(z)\chi_2.
$$
For the proof of \eqref{eee1} we insert exponentials as in \eqref{hcadoua16} and use \eqref{hcadoua7} and \eqref{expdek}. For the proof of \eqref{eee20} we use the same idea of splitting the exponential decay as in \eqref{hcadoua116}.  Namely, if we choose $\x_0$ in the support of $\chi_1$ we can write: 
\begin{align}\label{finalaa2}
\chi_1 W(H_D-i)^{-1}(H_D-z)^{-1}\chi_2&=\chi e^{\frac{\delta\eta }{r} |\x_0-\cdot |}\left\{ e^{-\frac{\delta\eta }{r} |\x_0-\cdot |} 
W(H_D-i)^{-1} e^{\frac{\delta\eta }{r} |\x_0-\cdot |} \right \} \\
&\cdot \left \{ e^{-\frac{\delta\eta}{2r} |\x_0-\cdot|} e^{-\frac{\delta\eta}{2r} |\x_0-\cdot|}(H_D-z)^{-1}
e^{{\frac{\delta\eta}{2r}|\x_0-\cdot|}}\right \} e^{-{\frac{\delta\eta }{2r}|\x_0-\cdot|}}\chi_2.\nonumber 
\end{align}
The second parenthesis contains a Hilbert-Schmidt operator, while the
other factors are bounded. The decay comes from the last factor.

Estimates \eqref{eee10} and \eqref{eee2} can be proved as the previous ones, taking into account that $z=-1$ can be used in our Combes-Thomas estimates (Proposition \ref{prop2}), since it belongs to the resolvent set of $H_D$; see also the proof of Proposition \ref{comm}.
\end{proof}

\subsection{Pointwise bounds on the Green's function}

The goal of this subsection is to prove a pointwise bound on the resolvent kernel $\Green(\cdot,\cdot,z)$ of $\left( H_D - z \right)^{-1}$, where the operator $H_D$ is defined as in the previous subsection.

\begin{proposition}\label{diamag}
 There are constants $C>0$ and $\delta>0$ such that for all $z\in\mathcal D$ and for all $\x,\x'\in\R^2$ one has
\begin{equation}
 \label{eq:diamag}
| \Green(\x,\x',z) | \leq C \frac{r^4}\eta \left(1+|\ln|\x-\x'||\right) e^{-\frac{\delta\eta}r |\x-\x'|} \,.
\end{equation}
\end{proposition}

We emphasize that the constants $C$ and $\delta$ can be chosen independently of $\A$.

\begin{proof}
We use the resolvent identity \eqref{eq:resid}. By the diamagnetic inequality, as in the proof of Corollary \ref{localizHS},
$$
\left| \Green(\x,\x',-1) \right| \leq C \left(1 +|\ln |\x-\x'||\right) e^{- \langle\x-\x'\rangle/2} \,.
$$
Moreover, since $\langle \x-\x'\rangle \leq \langle \x-\y\rangle+\langle \y-\x'\rangle$,
\begin{align*}
e^{{\frac{\delta \eta}{r } \langle \x-\x'\rangle}} \left| \int_{\R^2} d\y \, \Green(\x,\y,z)\, \Green(\y,\x',-1) \right|
\leq & \left( \int_{\R^2} d\y \, e^{{\frac{2\delta \eta}{r } \langle \x-\y\rangle}} |\Green(\x,\y,z) |^2 \right)^{1/2} \\
& \times \left( \int_{\R^2} d\y \, e^{{\frac{2\delta \eta}{r } \langle \y-\x'\rangle}} |\Green(\y,\x',-1) |^2 \right)^{1/2} \,.
\end{align*}
By \eqref{azecea14} and \eqref{hcadoua4}, both factors are finite and their product is bounded by a constant times $r^2/\eta$. This proves \eqref{eq:diamag}, even with $|\x-\x'|$ in the exponential replaced by $\langle\x-\x'\rangle$.
\end{proof}


\subsection{Pointwise bounds on the derivative of the Green's function}

In this subsection we specialize to the case of a constant magnetic field and an operator defined on all of $\R^2$. We shall prove

\begin{proposition}\label{kernelder}
 There are $\delta, C>0$ such that for any $\A\in C^1(\R^2)$ with $\curl\A$ constant and for any $z\in\mathcal D$, the kernel $\Green_0(\cdot,\cdot,z)$ of $((-i\nabla-\A)^2-z)^{-1}$ satisfies
$$
\left| (-i\nabla_\x-\A(\x)) \Green_0(\x,\x',z) \right| \leq C \frac{r^6}{\eta} |\x-\x'|^{-1}  e^{-\frac{\delta\eta}{r}|\x-\x'|} \,.
$$
\end{proposition}

We emphasize that the constants $C$ and $\delta$ can be chosen
independently of the value of the (constant) magnetic field.

\begin{proof}
By gauge covariance of the kernel we may assume that $\A(\x) = \frac{b}{2} (-x_2,x_1)$ for some $b\in\R$. Since the case $b=0$ is well-known and since the case $b<0$ can be reduced to the case $b>0$ via complex conjugation, we may also assume that $b>0$.

 According to Mehler's formula, the heat kernel of $H=(-i\nabla-\A)^2$ 
is given by
\begin{align}
  \label{eq:2}
 \exp(-tH)(\x,\x') = \frac{b}{4\pi \sinh (bt)} \exp(-\frac{ib \x\wedge\x'}2) \exp(-\frac{b|\x-\x'|^2}{4\tanh(bt)}) \,. 
\end{align}
Hence
$$
(-i\nabla_\x-\A(\x)) \exp(-tH)(\x,\x') = \frac{b}{2} \exp(-tH)(\x,\x')
\left( \frac{i (\x-\x')}{\tanh(bt)} - (\x-\x')^\bot \right) \,,
$$
where $(\x-\x')^\bot = (-(x_2-x_2'),x_1-x_1')$, and
$$
\left|(-i\nabla_\x-\A(\x)) \exp(-tH)(\x,\x')\right| = \frac{b |\x-\x'|}{2} \sqrt{\tanh^{-2}(bt) +1}
\left|\exp(-tH)(\x,\x')\right| \,.
$$
Using that $a\geq \tanh a$ for all $a\geq 0$ and that
$$
c:= \sup_{a\geq 0} \frac{a^2}{\sinh a} \sqrt{\tanh^{-2}a +1} <\infty \,,
$$
we find that
\begin{equation*}
 \label{eq:heatderiv}
\left|(-i\nabla_\x-\A(\x)) \exp(-tH)(\x,\x')\right| \leq \frac{c |\x-\x'|}{8\pi t^2} \exp(-\frac{|\x-\x'|^2}{4t}) \,.
\end{equation*}
We will compare this bound with
$$
\nabla_\x \exp(t\Delta)(\x,\x') = -\frac{\x-\x'}{8\pi t^2} \exp(-\frac{|\x-\x'|^2}{4t}) \,.
$$
{}From the last two relations we deduce that for any $z$ with $\Re z<0$
\begin{align*}
 \left|(-i\nabla_\x-\A(\x)) \Green_0(\x,\x',z)\right| 
= & \left|\int_0^\infty dt\, e^{tz} (-i\nabla_\x-\A(\x)) \exp(-tH)(\x,\x') \right| \\
\leq & c |\x-\x'| \int_0^\infty dt\, e^{t \Re z} \frac{1}{8\pi t^2} \exp(-\frac{|\x-\x'|^2}{4t}) \\
= & c \left|\nabla_\x (-\Delta-\Re z)^{-1}(\x,\x') \right| \,.
\end{align*}
Using the explicit form of $\nabla (-\Delta-\Re z)^{-1}$ as convolution with a Bessel function and the properties of that function we deduce that for some $C>0$
$$
\left|(-i\nabla_\x-\A(\x)) \Green_0(\x,\x',-1)\right| \leq C |\x-\x'|^{-1} e^{-|\x-\x'|/2} \,.
$$
(This is not optimal, but sufficient for our purposes.)

Now for arbitrary $z\in\mathcal D$ we write, using the resolvent identity,
\begin{align*}
(-i\nabla_\x-\A(\x)) \Green_0(\x,\x',z) = & (-i\nabla_\x-\A(\x)) \Green_0(\x,\x',-1) \\
& +(z+1) \int_{\R^2} d\y \, (-i\nabla_\x-\A(\x)) \Green_0(\x,\y,-1) \Green_0(\y,\x',z) \,.
\end{align*}
This, together with the estimate on $\Green_0(\y,\x',z)$ from Proposition \ref{diamag} and some simple estimates, completes the proof.
\end{proof}


\section{Existence and  continuity of certain integral kernels}\label{apendixD}

Let $H_D$ be as in the previous section, but defined on the whole
space $\R^2$.

\begin{proposition}\label{propexkern}~\\
For any  $f$ in $\mathbb S(\mathbb R)$, the integral kernel 
 $f(H_D)(\x,\x')$  of  $f(H_D)$ belongs to  $C^0(\mathbb
R^2\times \mathbb R^2)$.  
\end{proposition}
\begin{proof}
We start by establishing the following lemma.

\begin{lemma}\label{lemahoelder} ~\\
Let $U$ be  a compact set
in $\R^2$ and  $\chi$ be  a smooth and compactly supported function  which
  is equal to $1$ on $U$. Then, for sufficiently small $\epsilon >0$, we
  have that the operator $\chi(\cdot)(H_D+1)^{-1}e^{\epsilon\langle
    \cdot\rangle}$ maps $L^2(\R^2)$ into $H^2(\R^2)$. Moreover, for
  every $\psi\in L^2(\R^2)$, the function 
$\chi(\cdot) (H_D+1)^{-1}e^{\epsilon\langle
    \cdot\rangle}\psi$ is continuous on $U$. 
\end{lemma}
\begin{proof} Let us write: 
\begin{align}\label{exkern1}
\chi(\cdot)(H_D+1)^{-1}e^{\epsilon\langle
    \cdot\rangle}=\left \{ e^{\epsilon\langle \cdot\rangle}\chi(\cdot)\;
  (H_D+1)^{-1}\right\}\left \{ (H_D+1)e^{-\epsilon\langle
    \cdot\rangle}(H_D+1)^{-1}e^{\epsilon\langle \cdot\rangle}\right\}.
\end{align}
A standard argument shows that 
$\chi(\cdot)(H_D+1)^{-1}$ maps $L^2(\R^2)$ into $H^2(\R^2)$. The first factor 
$e^{\epsilon\langle \cdot\rangle}$ on the left remains bounded because $\chi$ has compact
support. Now if $\epsilon$ is small enough, the estimate
\eqref{expdek2} implies that $(H_D+1)e^{-\epsilon\langle
    \cdot\rangle}(H_D+1)^{-1}e^{\epsilon\langle \cdot\rangle}$ is
  bounded on $L^2(\R^2)$. This proves the first statement of the
  lemma. The second statement is a consequence of Sobolev embeddings. \end{proof}

We now can start the  proof of the proposition. If we introduce 
$\tilde{f}(t):=(t+1)^2f(t)$, then we write 
$$\chi(\cdot) f(H_D)\chi(\cdot)=\left\{\chi(\cdot)(H_D+1)^{-1}e^{\epsilon\langle \cdot\rangle}\right\}
e^{-\epsilon\langle \cdot\rangle}\tilde{f}(H_D)e^{-\epsilon\langle
  \cdot\rangle}\left\{e^{\epsilon\langle
    \cdot\rangle}(H_D+1)^{-1}\chi(\cdot)\right\}.$$
The operator $W:=e^{-\epsilon\langle \cdot\rangle}\tilde{f}(H_D)e^{-\epsilon\langle
  \cdot\rangle}$ is Hilbert-Schmidt, because we can write it as
$$W=\left \{e^{-\epsilon\langle \cdot\rangle}(H_D+1)^{-1}\right\} 
(H_D+1)\tilde{f}(H_D)e^{-\epsilon\langle
  \cdot\rangle}$$
and use the fact that $e^{-\epsilon\langle \cdot\rangle}(H_D+1)^{-1}$
is Hilbert-Schmidt (use \eqref{hcadoua4}). 

Hence $W$ is an integral operator, with an integral kernel
in $ L^2(\R^4)$.\\
Therefore we are left with the investigation of the continuity 
 on $U\times U$ of the integral kernel defined by:
$$f(H_D)(\x,\x'):=\int \int \Green(\y,\x;-1) e^{\epsilon \langle \y\rangle }
W(\y,\y')e^{\epsilon \langle
  \y' \rangle }\Green(\y',\x';-1)d\y\; d\y'.$$
Using \eqref{hcadoua4}, that $W$ is in $L^2$  and the Cauchy-Schwarz 
inequality, we obtain 
\begin{align}
  \label{eq:1}
  |f(H_D)(\x,\x')|\leq C \|W\|_{B_2},
\end{align}
uniformly in $\x,\x'\in U$,
where the constant $C$ only depends on the diameter of $U$. 
Since $W$ is Hilbert-Schmidt, we can approximate 
$W(\y,\y')$ (in Hilbert-Schmidt norm) with a finite sum of the type $\sum k_j(\y)h_j(\y')$ where
the $k$'s and $h$'s are $L^2$-functions. 
Applying the inequality \eqref{eq:1} to the difference of the two
resulting operators, we see that we can 
 approximate the function
$f(H_D)(\cdot,\cdot)$ (in the uniform norm on $U\times U$) with 
functions of the type 
$$\sum_{j=1}^n \{(H_D+1)^{-1}e^{\epsilon \langle \cdot \rangle } k_j\}(\x)
\overline{\{(H_D+1)^{-1}e^{\epsilon \langle \cdot \rangle }\overline{h}_j\}(\x')},$$
which from Lemma \ref{lemahoelder} are  continuous on compacts. 
Hence $f(H_D)(\cdot,\cdot)$ is  continuous on its variables on
$U\times U$. Since $U$ was arbitrary, the proof is over. 
\end{proof}

\noindent {\bf Acknowledgments.}  H.C. acknowledges support from the Danish 
F.N.U. grant {\it  Mathematical Physics}. 
S.F. was supported by the Lundbeck Foundation, the Danish Natural
Science Research Council and by the European 
Research Council under the European Community's Seventh Framework 
Program (FP7/2007--2013)/ERC grant agreement 202859.


\begin{thebibliography}{99}

\bibitem{Bol} C. Bolley.
\newblock Familles de branches de bifurcations dans les \'equations de
Ginzburg-Landau.
\newblock M2AN Math. Model. Numer. Anal. 25 (3) (1991), 307-335.


\bibitem{BoHe1} C.~Bolley and  B.~Helffer.
\newblock Application of semi-classical analysis to the asymptotic
 study of the supercooling field of a superconducting material.
\newblock Ann. Inst. H. Poincar\'{e} Anal. Non Lin\'{e}aire 58 (1991), no. 2, 189-233.

\bibitem{BHRS} Ph. Briet, P. D. Hislop, G. Raikov, E. Soccorsi.
\newblock Mourre estimates for a 2D magnetic quantum Hamiltonian on strip-like domains.
\newblock Cont. Math. 500, (2009), Amer. Math. Soc., Providence, 33-46. 

\bibitem{BRS} Ph. Briet, G. D. Raikov, E. Soccorsi.
\newblock Spectral properties of a magnetic quantum Hamiltonian on a
strip. 
\newblock Asymptotic Analysis 58 (2008), 127-155. 

\bibitem{CdV} Y. Colin de Verdi\`ere.
\newblock L'asymptotique de Weyl pour les bouteilles magn\'etiques.
\newblock Comm. Math. Phys. 105 (1986), no. 2, 327-335.

\bibitem{CT} J.M. Combes, L. Thomas.
\newblock Asymptotic behaviour of eigenfunctions for multiparticle 
Schr\"odinger operators.   
\newblock Comm. Math. Phys.  34  (1973), 251-270.

\bibitem{Cor0}  H.D. Cornean.
\newblock On spectral properties of Dirac or Schr\"odinger operators with magnetic field.
\newblock PHD Bucarest (1999).

\bibitem{Cor1} H.D. Cornean.
\newblock 
On the magnetization of a charged Bose gas in the canonical ensemble. 
\newblock Comm. Math. Phys. 212 (2000), no.1, 1-27.

\bibitem{Cor2} H.D. Cornean.
\newblock Magnetic response in ideal quantum gases: the thermodynamic
limit.
\newblock 
Markov Process. Relat. Fields 9 (2003), 547-566.

\bibitem{CN} H.D. Cornean and G. Nenciu.
\newblock On eigenfunction decay for two dimensional magnetic Schr\"odinger operators.
\newblock Comm. Math. Phys. 192 (1998), 671-685.



\bibitem{CN2} H.D. Cornean and  G. Nenciu.
\newblock  Two-dimensional magnetic
  Schr\"odinger operators: width of mini bands in the tight binding approximation.
 Ann. Henri Poincar{\'e} 1 (2000), no. 2, 203-222.

\bibitem{CNP} H.D. Cornean, G. Nenciu and T.G. Pedersen. 
\newblock The Faraday effect revisited: general theory.  
\newblock  J. Math. Phys. 47 (2006), no.1, 013511.

\bibitem{CN3} H.D. Cornean, G. Nenciu. 
\newblock The Faraday effect revisited: Thermodynamic limit. 
\newblock  J. Funct. Anal. 257 (2009), no. 7, 2024-2066.  


\bibitem {DaHe} M. Dauge and  B.~Helffer.
\newblock Eigenvalues variation I, Neumann problem for Sturm-Liouville
 operators.
\newblock J. Differential Equations 104 (1993), no. 2, 243-262.


\bibitem{DiSj} M.~Dimassi and  J.~Sj\"ostrand.
\newblock {\it Spectral Asymptotics in the semi-classical limit.}
\newblock London Mathematical Society. Lecture Notes Series 268. 
Cambridge University Press (1999).


\bibitem{FoHel1} S.~Fournais and  B.~Helffer.
\newblock Accurate eigenvalue asymptotics for magnetic Neumann
Laplacians. 
\newblock   Ann. Inst. Fourier 56 (2006), no. 1, 1-67. 

\bibitem{FHBook} S. Fournais and B. Helffer.
\newblock {\it Spectral methods in surface superconductivity.}
\newblock Progress in Nonlinear Differential Equations and Their
Applications 77. Birkh\"auser 2010.

\bibitem{FoKa} S. Fournais and A. Kachmar.
\newblock On the energy of bound states for magnetic Schr\"odinger operators,
\newblock J. Lond. Math. Soc. (2) 80 (2009), no. 1, 233-255.

\bibitem {Fra2} R.L. Frank.
\newblock On the asymptotic behavior of edge states for magnetic
Schr\"odinger operators.
\newblock Proc. Lond. Math. Soc. (3) 95 (2007), no. 1, 1-19.

\bibitem{FLW} R.L. Frank, M. Loss and T. Weidl.
\newblock  P\'olya's conjecture in the presence of a constant magnetic field.
\newblock J. Eur. Math. Soc. 11 (2009), 1365-1383.

\bibitem{He1} B.~Helffer.
\newblock {\it Semi-classical analysis for the Schr\"odinger operator and applications.}
\newblock Lecture Notes in Mathematics 1336. Springer Verlag 1988.


\bibitem {HeSyrie} B.~Helffer.
\newblock Introduction to semi-classical methods for the Schr\"odinger
operator with magnetic fields.
\newblock In: Aspects th\'eoriques et appliqu\'es de quelques EDP issues de la g\'eom\'etrie ou de la physique. A. El Soufi, M. Jazar (eds.), S\'eminaires et Congr\`es 17 (2009), 47-111.

\bibitem {HelMo1} B.~Helffer and   A.~Mohamed.
\newblock  Semiclassical analysis for the ground
state energy of a Schr\"odinger operator with magnetic wells. 
\newblock J. Funct. Anal. 138 (1996), no. 1, 40-81.

\bibitem {HelMo2} B.~Helffer and  A.~Morame.
\newblock Magnetic bottles in connection with superconductivity.
\newblock J.~Func. Anal. 185 (2010), no. 2, 604-680. (See Erratum  available at
http://mahery.math.u-psud.fr/~helffer/erratum164.pdf, 2005).

\bibitem{HeRo} B. Helffer and D. Robert.
\newblock Calcul fonctionnel par la transformation de Mellin et op\'erateurs admissibles.
\newblock J. Funct. Anal. 53 (1983), no. 3, 246-268.


\bibitem{He-Sjhva} B.~Helffer and J.~Sj{\"o}strand.
\newblock Equation de Schr\"odinger avec champ magn{\'e}tique et {\'e}quation de Harper.
\newblock  Springer Lecture Notes in Phys.  No. 345  (1989), 118-197.

\bibitem{HeSj}  B.~Helffer and J.~Sj{\"o}strand.
\newblock  On diamagnetism and the de Haas-Van Alphen  effect,
 \newblock Annales de l'IHP, section Physique th\'eorique 52 (1990), 303-375. 

\bibitem{HorSmil}  K.~Hornberger and U.~Smilansky.
\newblock Magnetic edge states.
\newblock Physics reports  367 (2002), no. 4, 249-285.

\bibitem{Iv} V.Ja. Ivrii.
\newblock Second term of the spectral asymptotic expansion of the Laplace - Beltrami operator on manifolds with boundary.
\newblock Functional Anal. Appl. 14 (1980), no. 2, 25-34.

\bibitem{Ku} H.~Kunz.
\newblock Surface orbital magnetism.
\newblock J. Stat. Phys. 76 (1994), no. 1/2, 183-207.

\bibitem{LSY} E.H. Lieb, J.P. Solovej and J. Yngvason.
\newblock Asymptotics of heavy atoms in high magnetic fields. II. Semiclassical regions.
\newblock Comm. Math. Phys. 161 (1994),  no. 1, 77-124.

\bibitem{Persson} M. Persson.
\newblock Eigenvalue asymptotics of the even-dimensional exterior Landau-Neumann Hamiltonian.
\newblock Adv. Math. Phys. (2009) Art. ID 873704, 15 pp.

\bibitem{Push} A. Pushnitski and G. Rozenblum.
\newblock Eigenvalue clusters of the Landau Hamiltonian in the exterior of a compact domain.  
\newblock Doc. Math. 12 (2007), 569-586 

\bibitem{SaVa} Yu. Safarov and D. Vassiliev.
\newblock The asymptotic distribution of eigenvalues of partial differential operators.
\newblock Translations of Mathematical Monographs, 155. Amer. Math. Soc., Providence, RI, 1997.

\bibitem{So1} A.V. Sobolev.
\newblock Quasi-classical asymptotics of local {R}iesz means for the
              {S}chr\"odinger operator in a moderate magnetic field.
\newblock Ann. Inst. H. Poincar\'e Phys. Th\'eor. 62 (1995), no. 4, 325-360.

\bibitem{So2} A.V. Sobolev.
\newblock The quasi-classical asymptotics of local {R}iesz means for the
              {S}chr\"odinger operator in a strong homogeneous magnetic
              field,
  \newblock Duke Math. J. 74 (1994), no. 2, 319-429.

\bibitem{So3} A.V. Sobolev.
\newblock Quasi-classical asymptotics for the Pauli operator.
\newblock Comm. Math. Phys. 194 (1998), no. 1, 109-134.

\bibitem{Ta} H. Tamura.
\newblock Asymptotic distribution of eigenvalues for {S}chr\"odinger operators with magnetic fields.
\newblock Nagoya Math. J. 105  (1987), 49-69.

\end{thebibliography}
\end{document}